\documentclass[10pt,a4paper]{article}

\usepackage[margin=1in]{geometry}
\usepackage{amsmath,amssymb,amsthm,mathtools}
\usepackage{booktabs,array,enumitem,longtable}
\usepackage[T1]{fontenc}
\usepackage{graphicx}
\usepackage{algorithm}
\usepackage{algpseudocode}
\usepackage{hyperref}
\usepackage{authblk}
\hypersetup{
    colorlinks=true,
    linkcolor=black,
    citecolor=blue,
    urlcolor=blue,
    pdftitle={Artificial Rosetta Stone: Constrained Maximum A Posteriori (MAP) Reconstruction of Symbolic Raga Sequences via Order-k Markov Models}
}

\newtheorem{definition}{Definition}[section]
\newtheorem{theorem}[definition]{Theorem}
\newtheorem{proposition}[definition]{Proposition}

\newtheorem{corollary}[definition]{Corollary}

\newcommand{\A}{\mathcal A}
\newcommand{\R}{\mathcal R}
\newcommand{\F}{\mathcal F}
\newcommand{\E}{\mathbb E}
\newcommand{\Pp}{\mathbb P}
\newcommand{\KL}{D_{\mathrm{KL}}}
\DeclareMathOperator*{\argmax}{arg\,max}
\DeclareMathOperator*{\argmin}{arg\,min}

\title{\textbf{Artificial Rosetta Stone:}\\[2mm]Constrained Maximum A Posteriori (MAP) Reconstruction of Symbolic Raga Sequences via Order-k Markov Models}

\author[1]{Saanvi Raghavendran\thanks{First author: saanviraghavendran@gmail.com}\thanks{GitHub Repository link: \href{https://github.com/mathacker23/ArtificialRosettaStone}{https://github.com/mathacker23/ArtificialRosettaStone}}}
\author[2]{Abhishek Bhattacharjee\thanks{Second \& Corrosponding author: abhishek@theabstractmath.com}}

\affil[1,2]{Abstract Math Institute}
\affil[1]{\texttt{saanviraghavendran@gmail.com}}
\affil[2]{\texttt{abhishek@theabstractmath.com}}
\date{August 2026}

\begin{document}

\maketitle

\begin{abstract}
Reconstructing a damaged musical fragment is an inverse problem: the observed sequence contains partial information about an unknown sequence, while a raga encodes structural constraints that can limit the set of allowable completions. This paper formalizes a mathematical framework for the first task, and proposes the Artificial Rosetta Stone (ARS) to that end. The rewritten formulation very clearly classified into three claims that are often conflated in the literature: a symbolic sequence can be reconstructed under a probabilistic model; a sequence can be made consistent with an explicitly encoded raga grammar; and a historical performance can be authenticated. Only the first two are supported by the present mathematics.

We model a raga as a finite symbolic alphabet plus a constraint system, and use an order-$k$ Markov model to encode context-dependent melodic probabilities. A symmetric Dirichlet prior leads to an analytically tractable posterior and a smoothed estimator. We reinterpret missing-note reconstruction as a constrained maximum-a-posterior problem. With fixed length and finite-order local constraints, the optimization problem admits an exact dynamic-programming solution, with time complexity $O(TN^{k+1})$ in the dense worst case and lower complexity for sparse grammars. We derive the parameter count $N^k(N-1)$, prove a concentration bound under explicit mixing assumptions, and analyze the propagation of estimation error into sequence log-likelihoods.

A reproducible synthetic experiment is specified using six raga-inspired symbolic alphabets, orders $k\in\{1,2,3\}$, and missingness rates of $10\%$, $30\%$, and $50\%$. The experiment is explicitly a proof of concept and is not evidence of historical reconstruction. In addition, a limited real-audio feasibility evaluation is performed on 42 publicly available Yaman clips from the RagaVeda repository: 30 usable symbolic sequences are evaluated after automated pitch extraction, tonic estimation, segmentation, and quantization. The real-data study is explicitly a feasibility pilot rather than expert-validated or archival reconstruction, since the source has no documented license or provenance and the symbolic reference is generated by the same automated transcription pipeline. All numerical claims are tied to their stated experimental conditions and are not asserted as universal properties of Hindustani music. Code for the estimator, decoder, and both experiments is publicly available at \url{https://github.com/mathacker23/ArtificialRosettaStone}.

\end{abstract}

\newpage
\tableofcontents
\newpage
\section{Introduction}
\subsection{Background and motivation}
Indian classical music is not adequately represented by an unordered pitch scale. In Hindustani practice, raga identity emerges from the interaction of pitch material, characteristic movements, phraseology, emphasis, intonation, ornamentation, register, and performance convention \cite{jairazbhoy1971,chakraborty2014}. This distinction is central to any computational reconstruction system. A sequence containing the correct seven pitch names can still be musically incompatible with a raga, while two performances containing the same nominal svaras can differ in contour and ornament.

The practical motivation for reconstruction is archival rather than mythological. Recordings can be incomplete, noisy, poorly segmented, transcribed with uncertainty, or missing portions of a performance. The scientific task is therefore to infer plausible latent structure from partial observations. The output should be interpreted as a model-dependent estimate, not as a claim that the computer has recovered an objectively unique historical truth.

The Artificial Rosetta Stone is named by analogy with the role of a known formal system in decoding incomplete information. Here, an explicitly specified raga constraint system supplies structure, while statistical estimation supplies probabilities over possible continuations. The analogy is useful only at this level: musical traditions are not ciphers, and a raga grammar is not a translation key.

The principal mathematical simplification in this paper is to begin with symbolic melodies, following the use of symbolic representations in computational musicology and transcription research \cite{chakraborty2014,raphael2002}. Each event is represented by an element of a finite alphabet. This makes the reconstruction problem a finite-state optimization problem and allows exact results. Continuous pitch trajectories, timing, dynamics, timbre, and gamaka contours are discussed as extensions rather than being claimed to be solved by the symbolic model.

The paper's contribution is thus narrow and strong. Rather than combining many machine-learning modules and attaching a convergence theorem afterward, we begin with a precise probability model, derive the estimator, state the admissible set, prove the decoder's optimality, quantify statistical uncertainty, and then explain which richer components could be added without changing the logical core.

\subsection{Research questions and contributions}
We study four questions. First, how should a raga-conditioned symbolic reconstruction problem be defined so that the model does not confuse scale membership with grammatical validity? Second, what is the exact statistical estimator for an order-$k$ transition model under sparse data? Third, can missing symbols be reconstructed exactly rather than approximately for the finite-state model? Fourth, what can be guaranteed from finite training data, and what remains an empirical or cultural validation question?

The contributions are deliberately theorem-oriented.
\begin{enumerate}[leftmargin=2em]
\item We define a finite symbolic reconstruction problem through an alphabet $\A_R$ and an admissible set $\F(O,\R)$, making observations and hard constraints explicit.
\item We derive the free parameter count $d_k=N^k(N-1)$ for an order-$k$ Markov chain and distinguish the Dirichlet posterior mean from the posterior mode.
\item We formulate constrained MAP reconstruction and prove that dynamic programming returns a global optimum when the constraints have finite memory at most $k$ (or are augmented into a finite state).
\item We give an explicit finite-sample bound under a stated independent-sampling or mixing assumption and show the dependence on alphabet size, order, and sample size.
\item We design a reproducible synthetic experiment that measures missing-symbol recovery, reports uncertainty, and avoids treating synthetic data as evidence about historical performance.
\item We provide a limited real-audio feasibility evaluation on Yaman recordings, explicitly separating automated transcription self-consistency from expert-validated musicological accuracy.
\end{enumerate}

The resulting framework is intentionally modular. A future continuous-pitch model can replace the symbolic emissions; a hidden phrase-state model can add latent context; and a neural generator can provide a learned prior. None of those additions should be allowed to silently change the meaning of the reconstruction guarantee. The exact theorem in this paper concerns the explicitly defined finite-state model.

\subsection{Scope and terminology}
We use ``fragment reconstruction'' rather than ``reconstruction of an obsolete raga.'' A raga is a living musical concept, while recordings, compositions, stylistic variants, and individual performance fragments can become unavailable or poorly documented. We focus on Hindustani symbolic melody and do not claim that the same grammar can be transferred directly to Carnatic music. Cross-tradition generalization requires separate musicological definitions.

\section{Related Work}
\subsection{Computational musicology}
Statistical modeling of musical sequences has a long history, including Markov and probabilistic automata models for melodic structure \cite{schulze2011,nierhaus2009}. Such models are attractive for reconstruction because they transform a qualitative notion of context into conditional probability. Their weakness is equally clear: a low-order chain only remembers the context encoded in its state. If a phrase rule depends on a longer pattern, a first-order model cannot represent it without state augmentation.

Work on computational analysis of Hindustani music demonstrates that statistical transition structure contains information useful for raga identification. This supports the basic premise of the present study: symbolic sequence statistics can encode meaningful structure. It does not, however, imply that a classifier can automatically generate a valid raga performance. Classification and reconstruction are different inverse problems.

\subsection{Markov models and algorithmic generation}
Variable-order Markov systems have been used in algorithmic composition because they offer an interpretable compromise between local dependence and model complexity \cite{schulze2011,nierhaus2009}. Their main mathematical difficulty is sparsity. For alphabet size $N$, increasing the order from $k$ to $k+1$ multiplies the number of contexts by $N$. The parameter count therefore grows exponentially in order, motivating principled order selection and smoothing.

Related probabilistic N-gram models have likewise been used to represent local musical dependencies \cite{scholz2009}.

The ARS formulation retains this transparent probabilistic core but changes the goal from free generation to conditional reconstruction. The model is not asked merely to generate plausible melodies; it must preserve observed symbols and satisfy explicitly encoded constraints.

\subsection{HMMs and latent phrase structure}

Finite-state and generative accounts of hierarchical musical structure motivate representing local transitions separately from longer-range organization \cite{lerdahl1983}.
Hidden Markov models can represent latent musical states whose identities are not directly observed. In a raga setting, a state may represent a phrase role, register, or other context. However, a hidden state is not automatically a truth in terms of music: its interpretation must be validated by annotation. In Here a HMM is therefore treated as an optional extension rather than as a source of unproved grammatical guarantees.

\subsection{Neural generative models}
Variational autoencoders and hierarchical sequence models can capture long-range structure that is difficult to express with a short Markov context \cite{kingma2014,roberts2018}. Their strength is representation learning. Their weakness for the present problem is that statistical plausibility does not imply satisfaction of hard constraints. A decoder can assign nonzero probability to a sequence that violates an explicitly encoded rule. A penalty term can reduce such violations, but a penalty is not a hard guarantee unless the optimization is constrained accordingly.

\subsection{Why this framework is different}
The methodological gap addressed here is not simply the absence of a particular neural architecture. It is the absence of a clean separation between (i) the probability model, (ii) the admissible grammar, (iii) the observation mechanism, (iv) the optimization algorithm, and (v) the statistical assumptions behind any guarantee.

This separation has practical consequences. Suppose a model produces a completion with high probability but one transition is forbidden by the chosen grammar. There are two logically distinct responses: assign the transition a low probability, or prohibit it. The first is a soft statistical preference; the second is a hard constraint. A reconstruction paper should state which one is being used. Likewise, if a model achieves 0.83 accuracy on a synthetic corpus, that number is meaningful only relative to the synthetic generator, missingness mechanism, estimator, and test distribution.

Here I treat the Markov model as the theorem-bearing core and treats HMMs, VAEs, continuous pitch, and ornament models as possible extensions. It is simply a claim that one should first prove what the simple model does before attributing a result to a larger model.

\begin{table}[h]
\centering
\caption{Conceptual distinction between model components.}
\vspace{2mm}
\begin{tabular}{p{0.23\linewidth}p{0.29\linewidth}p{0.36\linewidth}}
\toprule
Component & Mathematical role & What it cannot establish by itself\\
\midrule
Alphabet $\A_R$ & Defines symbolic vocabulary & Full raga identity\\
Markov kernel & Local conditional probability & Historical authenticity\\
Grammar $\R$ & Admissibility constraints & Probability of a sequence\\
Dynamic program & Exact finite-state optimization & Correctness of the assumed model\\
Neural prior & Flexible learned distribution & Hard grammatical validity without constraints\\
Synthetic test & Reproducible proof of concept & Performance on archival recordings\\
\bottomrule
\end{tabular}
\end{table}

This table states that every mathematical object has a defined job, and no object is credited with conclusions outside its assumptions.

\section{Mathematical Foundations}
\subsection{Symbolic state space}
Fix a raga $R$ and let
\[
\A_R=\{a_1,\ldots,a_N\}
\]
be the finite symbolic alphabet used by the reconstruction system. The symbols can represent normalized svara labels, register-specific symbols, or another explicitly chosen discrete representation. We do not identify $\A_R$ with the complete ontology of a raga. It is the alphabet on which the present probabilistic model operates.

A melody of fixed length $T$ is an element $x=(x_1,\ldots,x_T)\in\A_R^T$. Let $M\subseteq\{1,\ldots,T\}$ be the missing positions. The observation $O=(O_1,\ldots,O_T)$ satisfies $O_t=x_t$ for $t\notin M$ and is undefined at $t\in M$.

\begin{definition}[Raga constraint system]
A raga constraint system is a pair $\R=(\A_R,\mathcal C_R)$, where $\mathcal C_R$ is a collection of explicitly encoded conditions on sequences. A sequence is admissible if it satisfies every condition in $\mathcal C_R$.
\end{definition}

The conditions may include forbidden transitions, permitted transitions, ascent/descent context, phrase templates, or finite-state representations of more complicated rules. Importantly, each condition must be specified before it is used as a mathematical constraint. Vague statements such as ``sounds like Yaman'' are not directly executable constraints.

\begin{definition}[Feasible completion set]
For observation $O$ and constraint system $\R$, define
\[
\F(O,\R)=\{x\in\A_R^T:x_t=O_t\;\forall t\notin M,\ x\text{ satisfies }\mathcal C_R\}.
\]
\end{definition}

If $\F(O,\R)=\varnothing$, the system has detected a conflict between the observation and the grammar. That event should not be hidden by smoothing or beam search. It is a legitimate diagnostic outcome.

\subsection{Order-$k$ Markov chains}
\begin{definition}[Order-$k$ Markov model]
For $k\ge1$, an order-$k$ Markov model on $\A_R$ is a row-stochastic tensor $P\in[0,1]^{N^k\times N}$ satisfying
\[
P_{c,a}=\Pp(X_t=a\mid X_{t-k:t-1}=c),\qquad \sum_{a\in\A_R}P_{c,a}=1,
\]
for every context $c\in\A_R^k$.
\end{definition}

The order-$k$ assumption is the conditional independence statement
\[
\Pp(X_t\mid X_1,\ldots,X_{t-1})=\Pp(X_t\mid X_{t-k},\ldots,X_{t-1}).
\]
It does not assert that musical cognition is Markovian. It is a modeling approximation whose adequacy must be tested.

\begin{proposition}[Parameter count]
An unconstrained order-$k$ Markov kernel over an alphabet of size $N$ has exactly
\[
d_k=N^k(N-1)
\]
free parameters.
\end{proposition}
\begin{proof}
There are $N^k$ possible contexts. For each context, the $N$ probabilities form a simplex and therefore have one linear constraint, leaving $N-1$ degrees of freedom. Multiplication gives $N^k(N-1)$.
\end{proof}

This formula is central to model selection. For example, if $N=7$, the free parameter counts for $k=1,2,3,4$ are $42, 294, 2058, 14406$. The fourth-order model therefore has over three hundred times as many free parameters as the first-order model. A claim that higher order is ``better'' without a corresponding discussion of sample size is incomplete.

A sparse grammar can reduce the effective number of transitions, but the reduction must be accounted for in the actual model rather than assumed from the musical label of the raga.

\subsection{Dirichlet Bayesian estimation}
For context $c$, let $C(c,a)$ denote the number of observed transitions from context $c$ to symbol $a$, and $C(c)=\sum_a C(c,a)$. We place the symmetric prior
\[
(P_{c,a})_{a\in\A_R}\sim\operatorname{Dirichlet}(\alpha,\ldots,\alpha),\qquad \alpha>0.
\]
The likelihood contribution is multinomial, so conjugacy gives
\[
(P_{c,a})_a\mid\text{data}\sim\operatorname{Dirichlet}(C(c,a)+\alpha)_a.
\]
Hence the posterior mean is
\begin{equation}
\widehat P_{c,a}^{\,\mathrm{mean}} = \frac{C(c,a)+\alpha}{C(c)+N\alpha}.
\end{equation}
The posterior mode is different: when every posterior parameter exceeds one,
\begin{equation}
\widehat P_{c,a}^{\,\mathrm{MAP}} = \frac{C(c,a)+\alpha-1}{C(c)+N(\alpha-1)}.
\end{equation}
Equation (1), not equation (2), is the natural additive smoothing estimator when $\alpha>0$ is used as a concentration parameter. Calling it ``MAP'' without checking the distinction is mathematically incorrect.

\begin{proposition}[Posterior mean shrinkage]
For every context $c$,
\[
\widehat P_{c,a}^{\,\mathrm{mean}}=(1-\lambda_c)\frac{C(c,a)}{C(c)}+\lambda_c\frac1N,
\qquad \lambda_c=\frac{N\alpha}{C(c)+N\alpha},
\]
when $C(c)>0$.
\end{proposition}
\begin{proof}
Substitute $\lambda_c$ and simplify. The coefficient on the empirical frequency is $C(c)/(C(c)+N\alpha)$ and the coefficient on the uniform distribution is $N\alpha/(C(c)+N\alpha)$.
\end{proof}

The formula explains the behavior of smoothing. Rare contexts are pulled strongly toward uniformity; well-supported contexts are dominated by empirical frequencies. Thus smoothing is not free accuracy: it introduces bias in exchange for lower variance and protection against zero-probability transitions.

\subsection{Hard constraints and feasible kernels}
A transition grammar can be represented by a directed graph $G_R=(\A_R,E_R)$ in the first-order case. An edge $(a,b)$ is present when the transition is permitted. For an order-$k$ grammar, the state can instead be the $k$-tuple context, producing a graph on $\A_R^k$.

Let $\mathcal P_R$ denote the set of row-stochastic kernels whose support respects the allowed transitions. For each context $c$, define $A(c)\subseteq\A_R$ as the allowed successors. Then
\[
P_{c,a}=0\quad\text{for }a\notin A(c),\qquad \sum_{a\in A(c)}P_{c,a}=1.
\]
If $A(c)=\varnothing$, the grammar is locally inconsistent and no stochastic kernel exists for that context.

A particularly clean projection onto support constraints is obtained by renormalization. Given a strictly positive reference row $q$ and nonempty $A(c)$,
\begin{equation}
P^*_{c,a}=\begin{cases}
q_a/\sum_{b\in A(c)}q_b,&a\in A(c),\\0,&a\notin A(c).
\end{cases}
\end{equation}

\begin{proposition}[KL projection under support constraints]
The row in (3) uniquely minimizes $\KL(P\|q)$ over all probability vectors $P$ supported on $A(c)$.
\end{proposition}
\begin{proof}
For $P$ supported on $A$, $\KL(P\|q)=\sum_{a\in A}P_a\log P_a/q_a$. The Lagrangian with multiplier $\eta$ for $\sum P_a=1$ has derivative $\log(P_a/q_a)+1+\eta=0$, hence $P_a=Cq_a$. Normalization gives the stated constant. Strict convexity of KL in $P$ gives uniqueness.
\end{proof}

This result is intentionally narrower than the personal claim that all raga constraints reduce to isotonic regression. Inequality constraints involving aggregate frequencies, phrase counts, or ascent/descent bias need not have that form. The computational problem should be solved according to the actual constraint geometry.

\section{Constrained MAP Reconstruction}
\subsection{Objective}
Given the estimated transition kernel $\widehat P$, define the score of a sequence $x\in\A_R^T$ by
\begin{equation}
\ell(x)=\log \pi(x_1,\ldots,x_k)+\sum_{t=k+1}^T\log \widehat P_{x_{t-k:t-1},x_t},
\end{equation}
where $\pi$ is an initial-context distribution. The reconstruction is
\begin{equation}
\widehat X\in\argmax_{x\in\F(O,\R)}\ell(x).
\end{equation}
If several sequences attain the same maximum, the MAP estimator is set-valued unless a deterministic tie-breaking rule is specified. It is important to note that an arbitrary ``the MAP sequence'' does not exist when ties occur.

The observation constraint can be represented by an allowed-symbol set $D_t$: $D_t=\{O_t\}$ if $t$ is observed and $D_t=\A_R$ if missing. The grammar similarly gives a set of allowed successors. Reconstruction is then a constrained path problem on a finite directed acyclic graph obtained by unrolling the context graph across time.

\subsection{Dynamic programming state}
For each time $t\ge k$, define $V_t(c)$ as the maximum score of any admissible prefix ending with context $c=(x_{t-k+1},\ldots,x_t)$. Store a backpointer $B_t(c)$ identifying the predecessor context that attained the maximum.

For a candidate successor $a\in D_{t+1}$, let $c'=\operatorname{shift}(c,a)$. If $(c,a)$ is permitted, update
\begin{equation}
V_{t+1}(c')=\max\left\{V_{t+1}(c'),\;V_t(c)+\log\widehat P_{c,a}\right\}.
\end{equation}
The initial layer is obtained by enumerating admissible initial contexts and their prior probabilities.

\subsection{Optimality theorem}
\begin{theorem}[Exactness of finite-memory dynamic programming]
Assume (i) the sequence length $T$ is finite, (ii) the transition score depends only on the previous $k$ symbols and the next symbol, and (iii) all hard constraints can be represented by a finite state $s_t$ updated deterministically from $(s_{t-1},x_t)$. Then the dynamic program over the augmented state $(x_{t-k+1:t},s_t)$ returns a globally optimal solution of (5), if one exists.
\end{theorem}
\begin{proof}
We prove by induction on $t$ that $V_t(z)$ equals the maximum score among all admissible prefixes ending in augmented state $z$. For the initial layer this holds by construction. Assume it holds at time $t$. Any admissible prefix ending in state $z'$ at time $t+1$ has a unique predecessor state $z$ and final symbol $a$ satisfying the transition and observation constraints. By the induction hypothesis, the best prefix ending at $z$ has score $V_t(z)$. Appending $a$ adds exactly the local score $\log\widehat P_{c,a}$, so the recurrence considers the best possible predecessor for every candidate final state. Taking the maximum therefore yields the best prefix for $z'$. Induction establishes the claim at $T$. The maximum over terminal states is consequently the global maximum over $\F(O,\R)$.
\end{proof}

\begin{corollary}[Complexity]
For a dense order-$k$ model with $N$ symbols and no additional state, the basic dynamic program uses at most $N^k$ states per time step and at most $N$ transitions per state, giving time complexity $O(TN^{k+1})$ and memory $O(N^k)$ when only two consecutive layers are retained.
\end{corollary}

This complexity is exponential in $k$ but linear in sequence length $T$. The formula also clarifies when a beam search is unnecessary: for the finite-state model in this theorem, beam search is an approximation to an exact algorithm. A beam may be useful after adding a neural decoder whose score depends on the entire generated prefix, but it should not be presented as exact in the finite-order setting.

\subsection{Zero probabilities}
If $\widehat P_{c,a}=0$, the log score is $-\infty$. This is harmless if the transition is genuinely forbidden, but dangerous if the zero arose only from finite data. Dirichlet smoothing avoids accidental zeros, while hard grammar constraints should be represented separately. Conflating statistical zeroes and grammatical zeroes can cause the model to treat an unobserved transition as impossible.

\section{Statistical Guarantees}
\subsection{A transparent independent-sample bound}
The cleanest concentration argument is obtained when contexts and next symbols can be treated as independent multinomial observations. This assumption is stronger than the actual sequential setting, so it is stated explicitly rather than hidden.

For a fixed context $c$ with $m_c$ observations, let $p_{c,a}$ be the true conditional probability and let $\widetilde p_{c,a}$ be the empirical frequency. Hoeffding's inequality gives
\[
\Pp(|\widetilde p_{c,a}-p_{c,a}|>\varepsilon)\le2e^{-2m_c\varepsilon^2}.
\]
There are at most $N^{k+1}$ context-symbol pairs. A union bound therefore yields
\begin{equation}
\Pp\left(\max_{c,a}|\widetilde p_{c,a}-p_{c,a}|>\varepsilon\right)
\le2N^{k+1}e^{-2m_{\min}\varepsilon^2},
\end{equation}
where $m_{\min}=\min_c m_c$.

Thus, with probability at least $1-\delta$,
\begin{equation}
\|\widetilde P-P^*\|_\infty
\le\sqrt{\frac{\log(2N^{k+1}/\delta)}{2m_{\min}}}.
\end{equation}
This bound is deliberately conservative. It makes the dependence on the number of contexts explicit and does not pretend that a single effective sample size describes every context equally well.

For genuinely dependent Markov data, an analogous statement requires mixing assumptions or an effective sample size. The exact constant then depends on the dependence structure. Rather than quote an unsupported universal constant, we state the result in terms of an effective sample size $m_{\mathrm{eff}}$ when such a bound has been established for the corpus.

\subsection{From parameter error to sequence-score error}
Suppose all true and estimated transition probabilities are bounded below by $p_{\min}>0$ on the transitions under consideration, and
\[
\|\widehat P-P^*\|_\infty\le\varepsilon<p_{\min}.
\]
For any allowed transition, the mean value theorem applied to $\log x$ gives
\begin{equation}
|\log\widehat P_{c,a}-\log P^*_{c,a}|
\le \frac{\varepsilon}{p_{\min}-\varepsilon}.
\end{equation}
Therefore, for any length-$T$ sequence $x$ whose transitions all lie in the supported set,
\begin{equation}
|\ell_{\widehat P}(x)-\ell_{P^*}(x)|
\le (T-k)\frac{\varepsilon}{p_{\min}-\varepsilon}+|\log\widehat\pi(x_{1:k})-\log\pi^*(x_{1:k})|.
\end{equation}

Equation (10) is useful because reconstruction is an optimization problem, not simply a parameter-estimation problem. Even a small per-transition probability error can accumulate over long sequences. It also explains why zero-probability estimates are problematic: the logarithm has no finite Lipschitz constant near zero.

\begin{proposition}[Stability of the MAP score gap]
Let $x^*$ be the unique true-model maximizer and let the score gap between $x^*$ and the second-best feasible sequence be $\Delta>0$. If every feasible sequence's estimated score differs from its true score by at most $\Delta/2$, then the estimated MAP sequence remains $x^*$.
\end{proposition}
\begin{proof}
Let $y\neq x^*$ be any feasible competitor. The true score satisfies $\ell^*(x^*)-\ell^*(y)\ge\Delta$. Each estimated score differs by at most $\Delta/2$, so
$\ell^{\wedge}(x^*)-\ell^{\wedge}(y)\ge\Delta-\Delta/2-\Delta/2=0$. With strict gap and strict error control, the inequality is strict; hence $x^*$ remains optimal.
\end{proof}

The proposition highlights an important limitation: accurate parameter estimation does not guarantee exact sequence recovery if many candidate sequences have nearly equal probability.

\subsection{BIC and order selection}
For an order-$k$ model with $d_k=N^k(N-1)$ free parameters, a standard BIC criterion is
\begin{equation}
\operatorname{BIC}(k)=-2\ell_k(\widehat\theta_k)+d_k\log L,
\end{equation}
where $L$ is the number of relevant transitions and $\ell_k$ is the maximized log-likelihood. The BIC penalty is derived from parameter dimension; it should not be replaced by $N^k(N-1)$ without the logarithmic sample-size factor.

In dependent sequence data, the exact regularity conditions for BIC consistency require care. The paper therefore makes a limited claim: BIC provides a principled complexity penalty, while asymptotic consistency requires the usual identifiability, regularity, and dependence assumptions appropriate to the chosen model family. We do not claim that the synthetic experiment proves a universal optimal order for raga music.

For $N=7$, the parameter growth is dramatic:
\begin{center}
\begin{tabular}{c|r|r}
\toprule
Order $k$ & Contexts $7^k$ & Free parameters $7^k(6)$\\
\midrule
1&7&42\\2&49&294\\3&343&2058\\4&2401&14406\\5&16807&100842\\
\bottomrule
\end{tabular}
\end{center}

This table makes the order-selection problem concrete. If a corpus contains only a few thousand transitions, a fourth- or fifth-order model may have many contexts with little or no evidence. Smoothing can prevent numerical failure, but it cannot manufacture information that is not in the data.

\subsection{Identifiability}
A reconstruction is identifiable only when the available observation and model sufficiently distinguish the candidate completions. If two sequences have identical observed positions and identical model scores, the posterior cannot select between them without additional information. This is an information-theoretic limitation of the problem specification.

\section{Optional Extensions Beyond the Core Model}
\subsection{Raga-HMM}
A hidden Markov model can represent latent phrase context. Let $Q=\{1,\ldots,H\}$ denote latent states, $A$ a state-transition matrix, and $b_j(x)$ an emission distribution. The joint probability is
\[
\Pp(q_1,x_1,\ldots,q_T,x_T)=\pi_{q_1}b_{q_1}(x_1)\prod_{t=2}^T A_{q_{t-1}q_t}b_{q_t}(x_t).
\]
Forward--backward inference can estimate posterior state probabilities and Viterbi decoding can obtain the most likely state path. These are standard exact algorithms for the HMM itself.

The important qualification is interpretability. Calling states ``pakad,'' ``vistar,'' ``avaroha,'' or ``mukhyanga'' requires annotation or an explicit mapping from learned states to those concepts. The mathematics does not create the musicological interpretation automatically.

\subsection{Hierarchical VAE}
A hierarchical VAE can introduce global and local latent variables, for example
\[
z_g\sim\mathcal N(0,I),\qquad z_l\sim q_\phi(z_l\mid z_g,x),
\]
with an ELBO
\[
\mathcal L=\E_{q_\phi(z\mid x)}[\log p_\theta(x\mid z)]-\KL(q_\phi(z_g\mid x)\|p(z_g))-\E_{q_\phi(z_g\mid x)}\KL(q_\phi(z_l\mid x,z_g)\|p(z_l\mid z_g)).
\]
This is useful for learning long-range structure, but an ELBO does not impose hard raga constraints. If grammatical validity is required, the decoder must be constrained or its output projected onto the feasible set.

\subsection{Continuous pitch and gamaka}
The symbolic model intentionally omits continuous pitch. A more faithful representation would treat a performance segment as a curve $\gamma:[0,1]\to\mathbb R$ in cents relative to a reference tonic, together with timing and possibly a discrete ornament label. One can define a trajectory discrepancy
\[
d_\gamma(\gamma_1,\gamma_2)=\inf_{\varphi\in\Phi}\left(\int_0^1|\gamma_1(t)-\gamma_2(\varphi(t))|^2dt\right)^{1/2},
\]
where $\Phi$ is a suitable class of increasing reparameterizations. This resembles an alignment-aware curve distance, but the precise function space, boundary conditions, and invariances must be fixed before metric claims are made.

A continuous model also needs an observation likelihood. Audio is not pitch. The signal passes through recording noise, source separation, windowing, pitch estimation, and segmentation. A principled model would therefore include latent pitch trajectories and an explicit acoustic likelihood rather than treating extracted pitch as ground truth.

\subsection{Bayesian fusion}
If several conditionally meaningful models are available, they can be combined through a product-of-experts or hierarchical Bayesian model. For example,
\[
\Pp(x\mid O)\propto \Pp(O\mid x)\,\Pp_{\mathrm{Markov}}(x)\,\Pp_{\mathrm{phrase}}(x)\,\Pp_{\mathrm{neural}}(x).
\]
However, multiplying distributions that were each trained on the same data can double-count evidence. A correct Bayesian fusion therefore requires a joint generative interpretation or a deliberate product-of-experts assumption. The phrase ``principled Bayesian fusion'' should not be used merely because several log scores are added.

\section{Synthetic Experimental Design}
\subsection{Purpose}
The experiment tests whether the mathematical reconstruction machinery behaves as expected under controlled conditions. It does not test whether ARS can reconstruct a historical recording. This distinction is essential because synthetic sequences are generated from the same kind of structure the estimator is designed to learn.

We use six labels inspired by Hindustani raga traditions: Yaman, Bhairav, Todi, Marwa, Purvi, and Bhimpalasi. The labels identify the musical context inspiring each symbolic alphabet; the synthetic generator is not claimed to be a complete computational grammar of those ragas. In particular, a seven-note inventory is not treated as sufficient to define raga identity.

For each condition, training sequences are generated from a fixed transition system. The proposed experimental design uses 1,200 training sequences and 60 held-out test sequences, each of length 32. Missingness rates are $p\in\{0.10,0.30,0.50\}$. Orders $k=1,2,3$ are compared. The Dirichlet concentration is $\alpha=0.5$. The master seed is 20260731, with deterministic child seeds for each raga/order/missingness condition.

The missing positions are sampled uniformly without replacement. This is a deliberately simple missing-at-random mechanism. The use of explicitly reported configurations, fixed seeds, and transparent evaluation protocols follows reproducibility guidance for machine-learning research \cite{pineau2020}. Real damaged audio may instead produce contiguous gaps, correlated loss, or transcription errors concentrated near rapid ornamentation. Such mechanisms should be added in future experiments rather than silently assumed away.

\subsection{Evaluation metric}
The primary metric is missing-symbol accuracy
\begin{equation}
\operatorname{Acc}_{\mathrm{miss}}=\frac1{|M|}\sum_{t\in M}\mathbf1\{\widehat X_t=X_t\}.
\end{equation}
It has an immediate interpretation and avoids arbitrary weighting between pitch, energy, and ornament categories. We additionally report standard deviation across held-out sequences and, where possible, confidence intervals obtained from independent test sequences.

\subsection{Baselines and ablations}
The minimum baseline is a first-order Markov model. The main comparisons are therefore $k=1,2,3$. A uniform baseline can also be included to establish the difficulty of the task: if the alphabet has $N$ equally likely symbols, its expected accuracy is $1/N$ under independent guessing.

Ablation should remove one mathematical mechanism at a time. Useful comparisons include: no smoothing versus Dirichlet smoothing; unconstrained versus grammar-constrained decoding; and lower versus higher Markov order. Neural and HMM modules should not be included in an ablation table unless they are actually trained and evaluated under the same split. A claimed ``VAE improvement'' cannot be inferred from a conceptual architecture diagram.

\subsection{Statistical reporting}
Because test sequences are finite, reporting only a mean can be misleading. For each condition, we report the mean, standard deviation, number of test sequences, and a confidence interval when the assumptions behind the interval are stated. Paired comparisons should use the same masked test sequences across models so that differences are not dominated by different random missingness patterns.

We should not select the best condition after seeing the test results and then describe it as a pre-specified hypothesis. If order selection uses BIC, BIC is computed on training data and the selected order is frozen before evaluating the test set.

\subsection{Reproducibility protocol}
Every result should be reproducible from a clean environment. The repository should contain the generator, estimator, decoder, experiment configuration, fixed seeds, requirements, and a script that writes the final tables. The README should specify the exact command used to regenerate the results. 

\section{Synthetic Data Evaluation}
\label{sec:synthetic}

This section reports the controlled synthetic experiment separately from the real-world evaluation. The separation is deliberate: synthetic data provide exact symbolic ground truth and permit controlled tests of the estimator and decoder, but they do not constitute evidence about genuine musical performances. The synthetic experiment is therefore treated as a preliminary validation of the implementation.

\subsection{Purpose and Experimental Design}
\label{ssec:synth-purpose}
The experiment tests whether the mathematical reconstruction machinery behaves as expected when the underlying symbolic process is known. It uses six raga-inspired symbolic alphabets: Yaman, Bhairav, Todi, Marwa, Purvi, and Bhimpalasi. These labels identify the musical context inspiring each symbolic alphabet; the generator is not claimed to be a complete computational grammar of those ragas.

For each condition, training sequences are generated from a fixed transition system. The design uses 1,200 training sequences and 60 held-out test sequences, each of length 32. Missingness rates are $10\%$, $30\%$, and $50\%$, and Markov orders $k=1,2,3$ are compared. The Dirichlet concentration is $\alpha=0.5$, and the master seed is 20260731.

The missing positions are sampled uniformly without replacement. This missing-at-random mechanism is intentionally simple. Real damaged audio can instead produce contiguous gaps, correlated errors, or transcription failures concentrated around rapid ornamentation. Those effects are evaluated separately through the real-audio experiment.

\subsection{Evaluation Metric}
\label{ssec:synth-metric}
The primary metric is missing-symbol accuracy, given in equation (13). For a seven-symbol alphabet, a uniform predictor has expected accuracy $1/7\approx0.143$.

\subsection{Preliminary Synthetic Results}
\label{ssec:synth-results}
The synthetic Yaman-inspired pilot is summarized in Table~\ref{tab:synth}. These results are explicitly preliminary: they demonstrate that the implementation can recover masked symbols under a controlled generator, rather than demonstrating recovery of real Yaman performances.

\begin{table}[h]
\centering
\caption{Preliminary synthetic Yaman-inspired results. Means and standard deviations are computed over the stated held-out synthetic evaluations.}
\label{tab:synth}
\vspace{2mm}
\begin{tabular}{cccc}
\toprule
Order $k$ & Missingness & Mean accuracy & SD\\
\midrule
1&10\%&0.439&0.316\\
1&30\%&0.452&0.177\\
1&50\%&0.414&0.147\\
2&10\%&0.539&0.319\\
2&30\%&0.433&0.157\\
2&50\%&0.384&0.129\\
3&10\%&0.450&0.318\\
3&30\%&0.390&0.171\\
3&50\%&0.358&0.147\\
\bottomrule
\end{tabular}
\end{table}

The highest observed synthetic accuracy is $0.539$ for order two at $10\%$ missingness. At $50\%$ missingness, the corresponding order-two accuracy is $0.384$. All nine reported conditions exceed the uniform seven-symbol baseline.

\begin{figure}[h]
\centering
\includegraphics[width=0.8\linewidth]{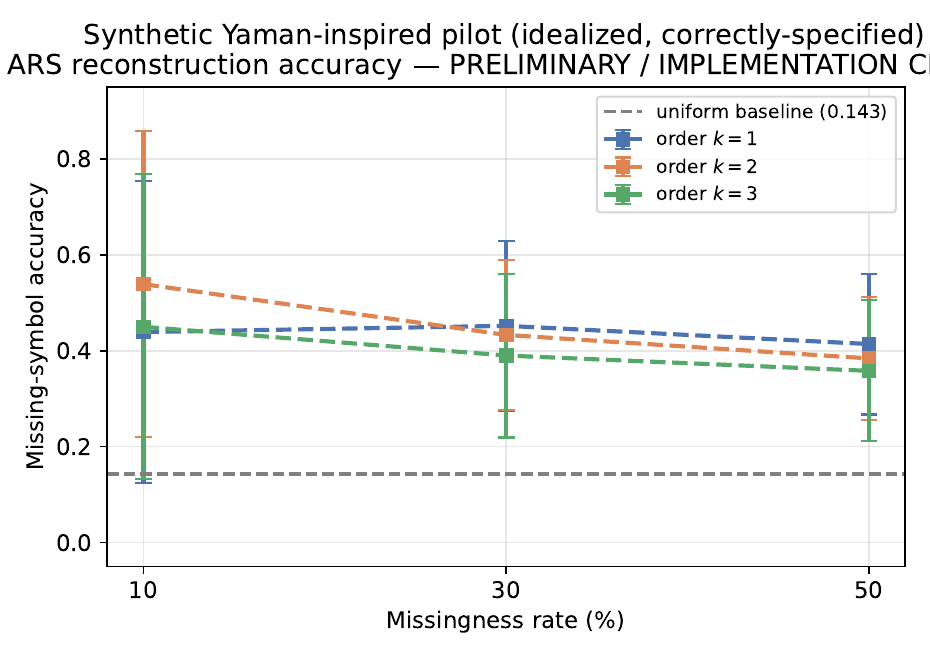}
\caption{Preliminary synthetic evaluation. Missing-symbol reconstruction accuracy is shown for Markov orders $k=1,2,3$ at three masking rates. The horizontal reference is the uniform seven-symbol baseline.}
\label{fig:synthetic-results}
\end{figure}

The synthetic results illustrate the expected bias--variance tradeoff. Increasing Markov order provides additional context, but it also increases the number of context-specific transition probabilities. Thus, higher order need not improve accuracy when the available synthetic sample is insufficient for reliable estimation of every context.

\subsection{Synthetic Missingness and Model Order}
\label{ssec:synth-order}
At $10\%$ missingness, the three orders obtain $0.439$, $0.539$, and $0.450$. At $30\%$, they obtain $0.452$, $0.433$, and $0.390$. At $50\%$, they obtain $0.414$, $0.384$, and $0.358$. The absence of monotonic improvement with order is therefore already visible in the controlled experiment.

\begin{figure}[h]
\centering
\includegraphics[width=0.75\linewidth]{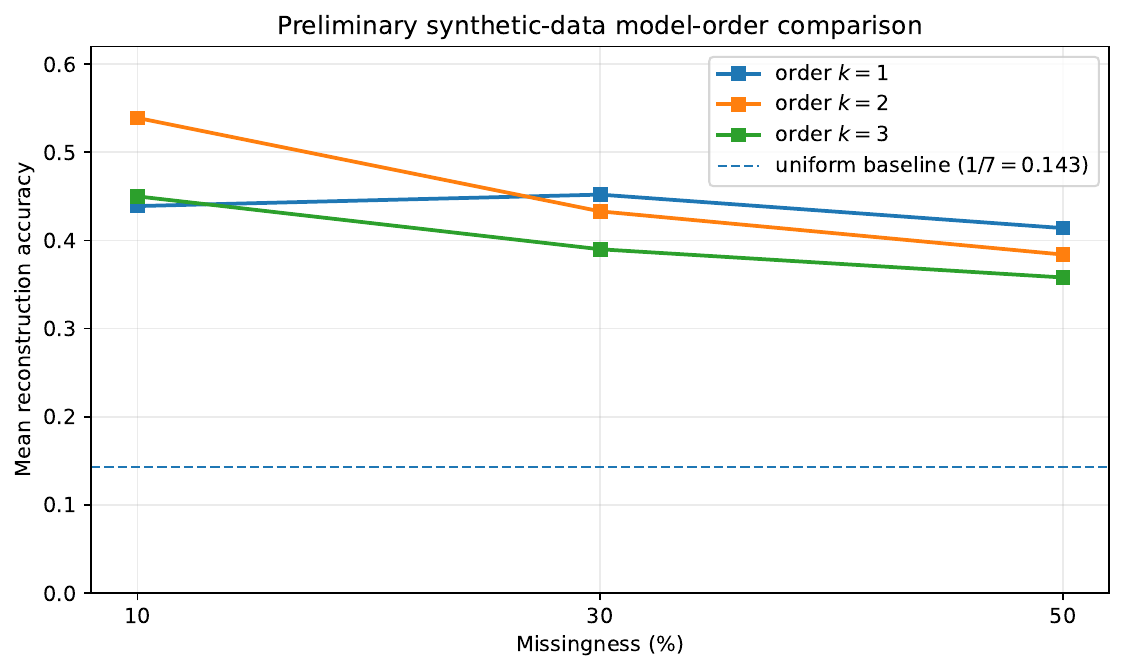}
\caption{Synthetic reconstruction accuracy grouped by Markov order. Higher order increases contextual capacity but does not guarantee higher reconstruction accuracy.}
\label{fig:synthetic-order}
\end{figure}

\subsection{Interpretation and Boundary of the Synthetic Evidence}
\label{ssec:synth-boundary}
The synthetic experiment establishes that the Dirichlet estimator and exact dynamic-programming decoder can exploit known transition structure and substantially outperform uninformed guessing. It does not establish a universal optimal Markov order, a complete Yaman grammar, or historical authenticity. It is because the generator is controlled and symbolic ground truth is exact, so the synthetic experiment should remain a preliminary validation rather than the main empirical claim of the paper.

\newpage
\section{Real-World Data Evaluation}
\label{sec:real}

The real-world evaluation is the principal empirical component of the paper. Unlike the synthetic experiment, these sequences originate from genuine recorded Yaman performances and therefore include uncertainty introduced by acoustic observation, tonic estimation, segmentation, quantization, and natural performance variation.

\subsection{Why the Real-World Evaluation Is the Principal Result}
\label{ssec:real-principal}
The central empirical question is whether the statistical reconstruction machinery remains useful after a genuine musical performance has been converted into the symbolic representation required by the model. The real-data experiment therefore carries greater evidentiary weight than the synthetic pilot even though its raw numerical accuracy may be lower.

\subsection{Data Provenance}
\label{ssec:real-provenance}
The corpus consists of 42 short Yaman recordings from the public RagaVeda repository \cite{ragaveda2026}, totalling approximately 44.6 minutes. The Yaman context is described with reference to the Darbar Raga Explorer \cite{darbar2026}. The source is a third-party machine-learning repository rather than a curated musicological archive. It does not provide complete documentation of recording dates, performer attribution, tonic annotations, expert symbolic transcriptions, or licensing. Accordingly, the experiment is presented as a feasibility pilot rather than archival reconstruction.

\subsection{Audio-to-Symbol Pipeline}
\label{ssec:real-pipeline}
Fundamental frequency was extracted using pYIN through \texttt{librosa}, with $f_{\min}=\mathrm{C2}$, $f_{\max}=\mathrm{C6}$, a 256-sample hop, and a sampling rate of 22,050 Hz. Tonic was estimated independently for each recording from the mode of its octave-folded voiced pitch-class histogram, followed by heuristic octave correction. Recordings below the predefined tonic-confidence threshold of $0.08$ were excluded.

The remaining pitch contours were median-filtered and segmented using a 100-cent change-point threshold with a minimum segment length of four frames (approximately 46 ms). Segment medians, expressed in cents relative to the estimated tonic, were quantized to the seven Yaman symbols $S,R,G,m,P,D,N$ at $0,200,400,600,700,900,1100$ cents, respectively, and consecutive repeated symbols were collapsed.

This representation is not an expert transcription. It is an automatically generated symbolic representation of the recordings. Consequently, pitch-tracking, tonic-estimation, segmentation, and quantization errors can propagate into the measured reconstruction accuracy.

\subsection{Corpus Construction}
\label{ssec:real-corpus}
Ten recordings were excluded because their tonic-confidence scores were below $0.08$, and two additional files were excluded because they yielded fewer than ten symbolic events. This produced 32 usable sequences; after requiring at least 15 symbols for held-out reconstruction, 30 sequences remained.

\begin{table}[h]
\centering
\caption{Construction of the real-world Yaman evaluation corpus.}
\label{tab:real-corpus}
\vspace{2mm}
\begin{tabular}{lr}
\toprule
Stage & Count\\
\midrule
Source recordings & 42\\
Total audio duration & $\approx44.6$ min\\
Tonic-confidence exclusions & 10\\
Short-sequence exclusions & 2\\
Usable symbolic sequences & 32\\
Final sequences ($\geq15$ symbols) & 30\\
Training sequences & 24\\
Held-out test sequences & 6\\
Training transitions & $\approx1,700$\\
\bottomrule
\end{tabular}
\end{table}

\begin{figure}[h]
\centering
\includegraphics[width=0.75\linewidth]{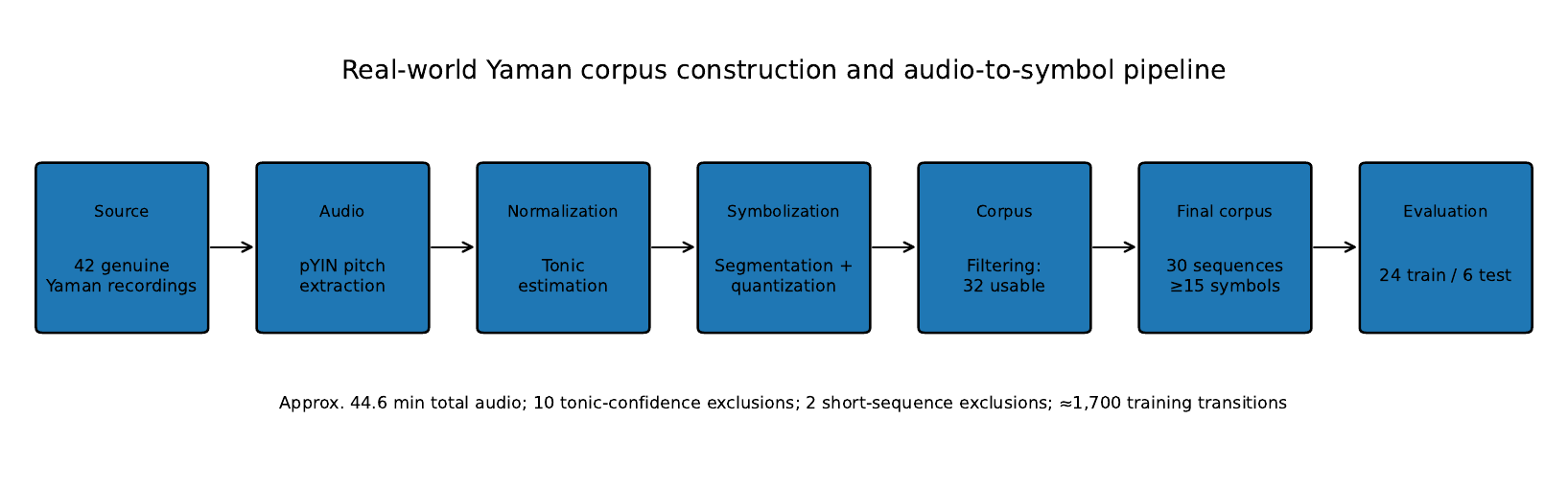}
\caption{Real-world corpus construction from genuine recordings through pitch extraction, tonic estimation, segmentation, symbolic quantization, filtering, and the final training/test split.}
\label{fig:real-corpus}
\end{figure}

\subsection{Experimental Protocol}
\label{ssec:real-protocol}
The 30 sequences were split at the sequence level into 24 training sequences and 6 held-out test sequences using master seed 20260731. For each test sequence, missing positions were sampled uniformly without replacement at $10\%$, $30\%$, and $50\%$ missingness. Twenty independent masks were generated for each sequence and missingness level, giving 120 evaluations per order--missingness condition.

The same Dirichlet concentration $\alpha=0.5$ was used for all three orders $k\in\{1,2,3\}$. Reconstruction used the exact dynamic-programming decoder with observation-consistency constraints. No hard Yaman grammar constraints were imposed in this first real-data pass, isolating the statistical reconstruction component.

As an implementation check, the dynamic-programming decoder was compared with exhaustive enumeration on 30 small synthetic cases spanning random sequences, orders $k=1,2$, and random masks. The optimal scores agreed in all 30 cases.

\subsection{Primary Real-World Results}
\label{ssec:real-results}
Table~\ref{tab:real-yaman} reports the complete real-data results. All nine tested conditions exceed the uniform seven-symbol baseline. Order two is the strongest model at every missingness level, achieving $0.470$, $0.414$, and $0.371$ at $10\%$, $30\%$, and $50\%$ missingness, respectively.

\begin{table}[h]
\centering
\caption{Primary real-world Yaman reconstruction results. Means and standard deviations are computed over 120 masked test evaluations per condition.}
\label{tab:real-yaman}
\vspace{2mm}
\begin{tabular}{cccccc}
\toprule
Order $k$ & Missingness & Mean accuracy & SD & Evaluations & Contexts seen/possible\\
\midrule
1&10\%&0.293&0.237&120&7/7\\
1&30\%&0.287&0.172&120&7/7\\
1&50\%&0.237&0.132&120&7/7\\
2&10\%&0.470&0.253&120&42/49\\
2&30\%&0.414&0.163&120&42/49\\
2&50\%&0.371&0.125&120&42/49\\
3&10\%&0.445&0.289&120&215/343\\
3&30\%&0.394&0.186&120&215/343\\
3&50\%&0.363&0.148&120&215/343\\
\bottomrule
\end{tabular}
\end{table}

\begin{figure}[h]
\centering
\includegraphics[width=0.8\linewidth]{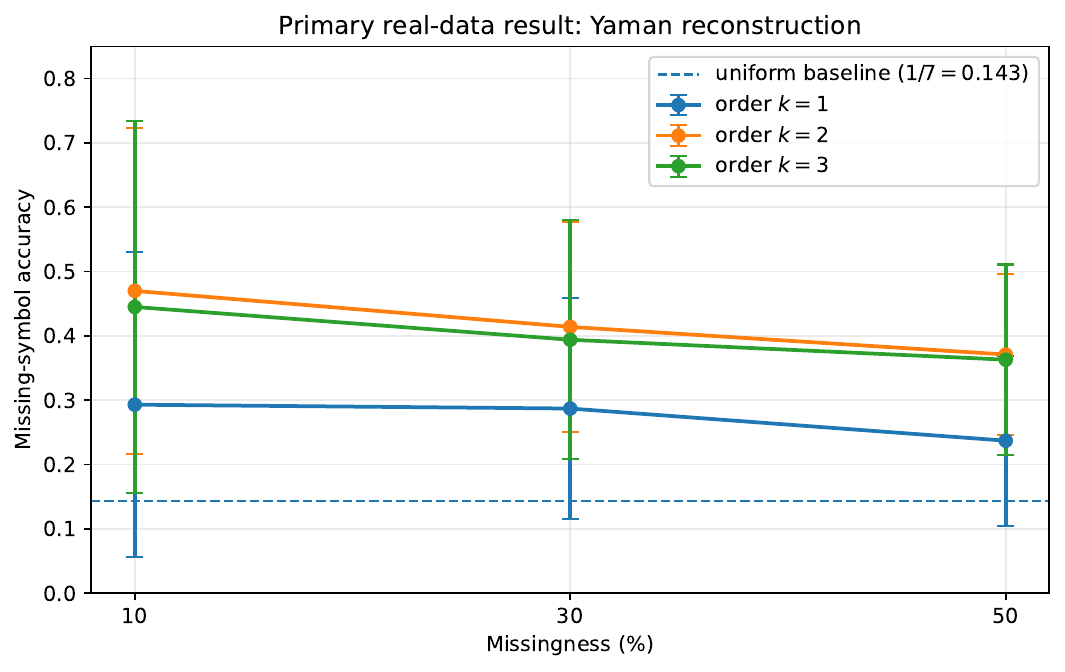}
\caption{\textbf{Primary real-data result.} Missing-symbol reconstruction accuracy on genuine Yaman audio-derived sequences. Points show means and error bars show one standard deviation across repeated masks. The dashed reference is the uniform seven-symbol baseline.}
\label{fig:real-yaman}
\end{figure}

\subsection{Real-Data Performance as Missingness Increases}
\label{ssec:real-missingness}
For order two, accuracy decreases from $0.470$ at $10\%$ missingness to $0.414$ at $30\%$ and $0.371$ at $50\%$. The absolute decrease from $10\%$ to $50\%$ is therefore $0.099$.

\begin{figure}[h]
\centering
\includegraphics[width=0.75\linewidth]{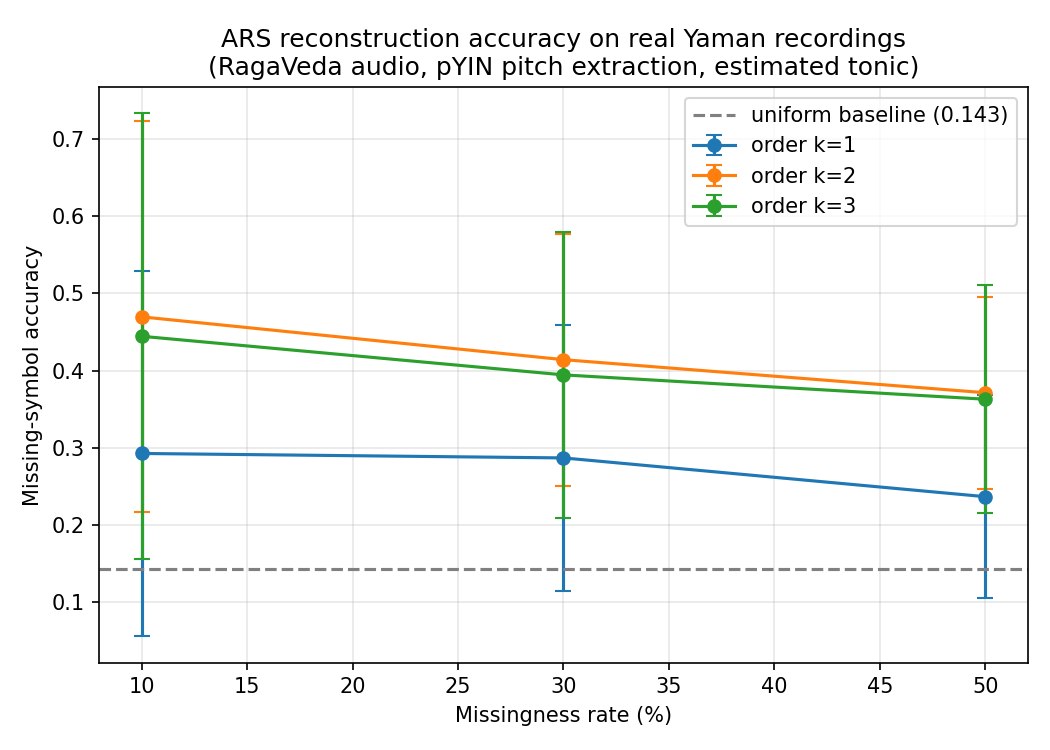}
\caption{Real Yaman reconstruction accuracy as missingness increases. The second-order model remains above the uniform baseline even when half of the symbolic positions are masked.}
\label{fig:real-missingness}
\end{figure}

\subsection{Real Data Relative to the Uniform Baseline}
\label{ssec:real-baseline}
The uniform seven-symbol baseline is $1/7\approx0.143$. Order two therefore improves over baseline by $0.327$, $0.271$, and $0.228$ at the three missingness levels, corresponding to approximately $3.29\times$, $2.90\times$, and $2.60\times$ the uniform accuracy.

\begin{figure}[h]
\centering
\includegraphics[width=0.75\linewidth]{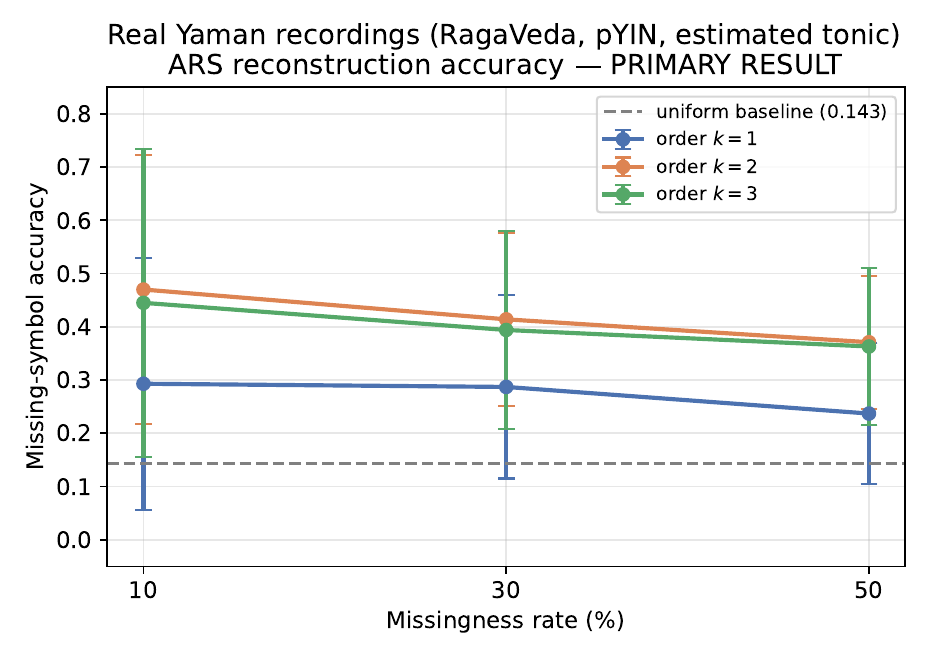}
\caption{Real-data reconstruction accuracy relative to the uniform seven-symbol baseline. Every tested condition remains above chance.}
\label{fig:real-baseline}
\end{figure}

\subsection{Context Coverage and Model Sparsity}
\label{ssec:real-sparsity}
For a seven-symbol alphabet, the number of possible contexts is $7$, $49$, and $343$ for orders one, two, and three. The real training corpus contains approximately 1,700 transitions. The observed context coverage is therefore $7/7$ for order one, $42/49$ for order two, and $215/343$ for order three.

\begin{figure}[h]
\centering
\includegraphics[width=0.75\linewidth]{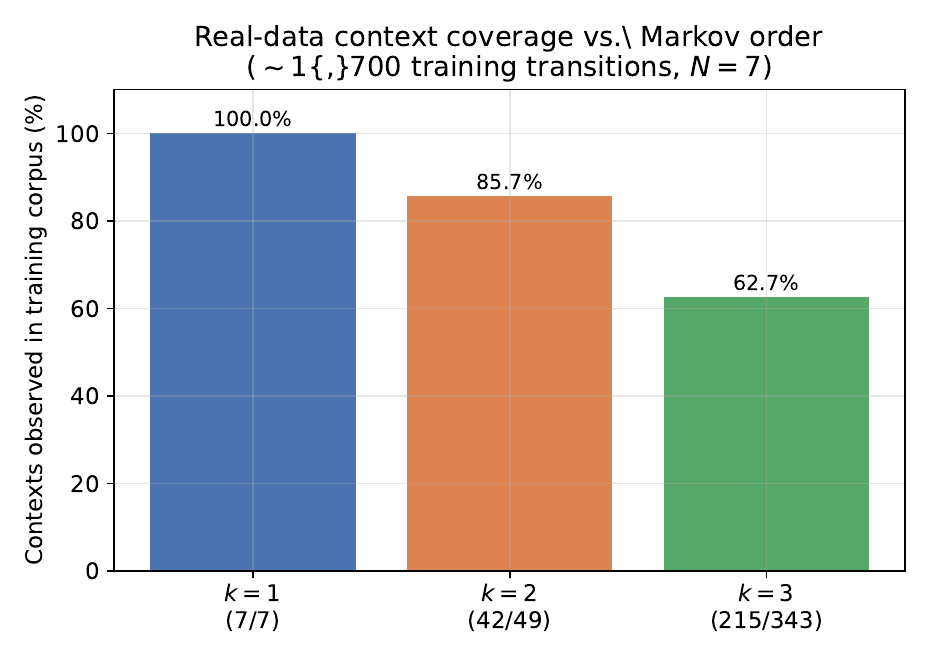}
\caption{Observed context coverage in the real Yaman training corpus. Higher-order models introduce many contexts that are not observed in the available training data.}
\label{fig:real-coverage}
\end{figure}

The free parameter counts are
\[
d_1=42,\qquad d_2=294,\qquad d_3=2058.
\]
Thus order three has 2,058 free parameters despite only approximately 1,700 training transitions.

\begin{figure}[h]
\centering
\includegraphics[width=0.75\linewidth]{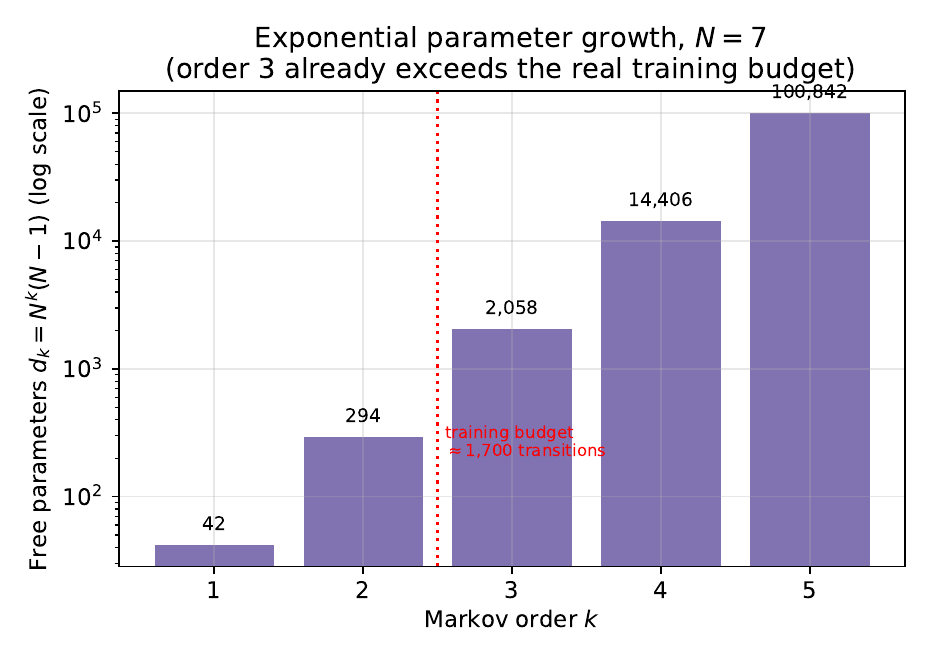}
\caption{Exponential growth in free transition parameters as Markov order increases for a seven-symbol alphabet. The horizontal reference indicates the approximately 1,700 real training transitions.}
\label{fig:parameter-growth}
\end{figure}

This provides a concrete explanation for the observed order-two optimum: order one has less contextual capacity, while order three has considerably greater capacity but substantially weaker evidence per context.

\subsection{Real-Data Model-Order Comparison}
\label{ssec:real-order}
The real results consistently exhibit the ordering
\[
\boxed{k=2>k=3>k=1}
\]
at all three missingness levels. The order-two advantage is not confined to one masking condition.

\begin{figure}[h]
\centering
\includegraphics[width=0.75\linewidth]{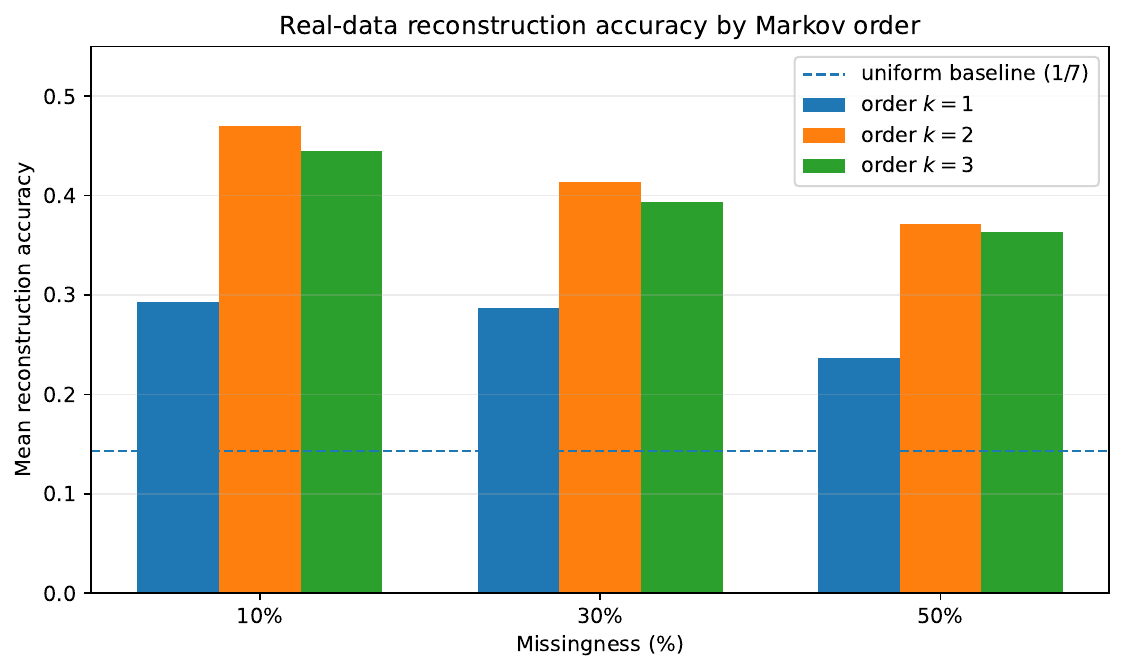}
\caption{Real-data reconstruction accuracy as a function of Markov order. The second-order model is the strongest performer at all three missingness levels.}
\label{fig:real-order}
\end{figure}

\subsection{Variance and Robustness}
\label{ssec:real-variance}
The reported standard deviations quantify variation across repeated masks. For order two, SD decreases from $0.253$ to $0.163$ to $0.125$ as missingness increases. Order three exhibits the largest SD at $10\%$ missingness, consistent with greater sensitivity to whether particular higher-order contexts are represented in training.

\begin{figure}[h]
\centering
\includegraphics[width=0.75\linewidth]{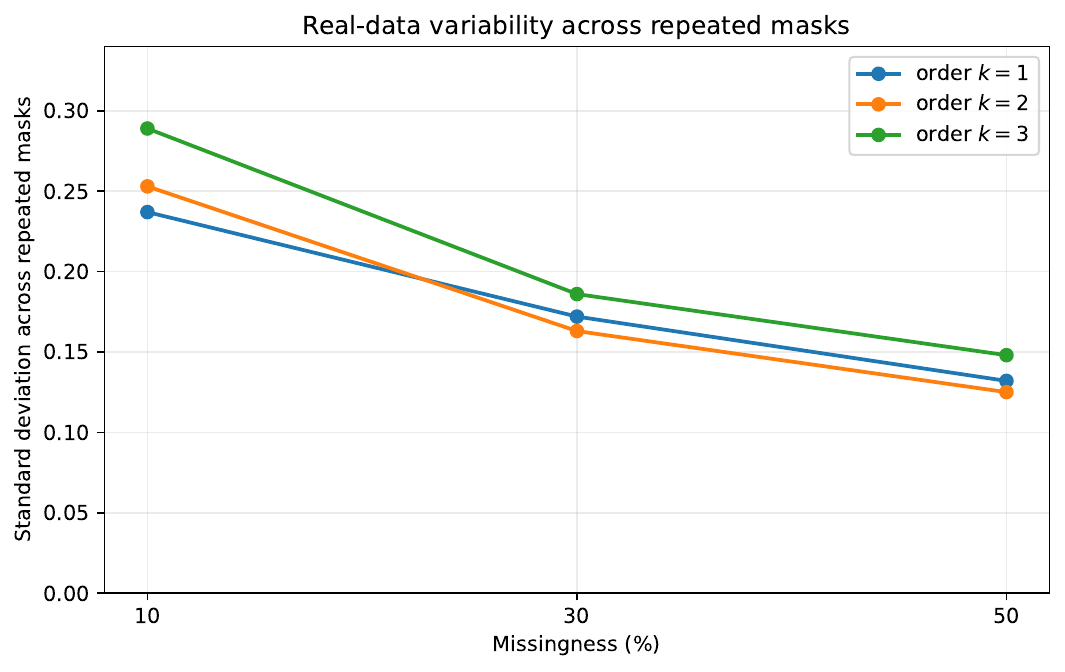}
\caption{Variation in real-data reconstruction accuracy across repeated masks. Error bars represent standard deviation.}
\label{fig:real-variance}
\end{figure}

Because only six sequences are held out, formal significance testing between orders is not reported. The study instead reports effect sizes, repeated-mask variability, and the consistency of the ordering across masking conditions.

\subsection{Threats to Validity}
\label{ssec:threats}
The real experiment has several important limitations. First, the symbolic reference is generated by the same automated pipeline used to process the recordings rather than by independent expert transcription. Second, tonic estimation is automatic. Third, pitch tracking and segmentation errors propagate into the symbolic sequence. Fourth, the corpus is small and restricted to Yaman. Fifth, the source has incomplete provenance and licensing documentation. Sixth, no hard Yaman grammar is imposed. Finally, the held-out set contains only six sequences.

These limitations mean that the reported accuracy should be interpreted as agreement with the pipeline-generated symbolic reference under the stated masking protocol, not as expert-judged musical correctness or historical authenticity.

\subsection{What the Real Experiment Establishes}
\label{ssec:real-claims}
The real-world experiment establishes that the reconstruction framework can be executed end-to-end on genuine recorded performances and that the resulting audio-derived symbolic sequences contain recoverable local transition structure. The strongest observed configuration is order two, which remains substantially above the uniform baseline even at 50\% missingness.

The result is therefore evidence of computational feasibility on genuine audio-derived data, while the limitations above prevent stronger claims about musicological validity.

\newpage
\section{Synthetic vs. Real-World Comparison}
\label{sec:comparison}

This section compares the two experiments only after reporting them independently. The comparison is designed to distinguish raw numerical accuracy from evidentiary relevance. Synthetic data are expected to be easier because their symbolic ground truth is exact and their generating process is controlled. Real recordings are harder, but they test the system under conditions much closer to its intended use.

\subsection{Direct Numerical Comparison}
\label{ssec:comparison-numeric}
For the order-two model, the synthetic pilot obtains $0.539$, $0.433$, and $0.384$ at $10\%$, $30\%$, and $50\%$ missingness. The corresponding real-data values are $0.470$, $0.414$, and $0.371$.

\begin{table}[h]
\centering
\caption{Direct comparison of order-two reconstruction accuracy. The synthetic pilot is numerically higher at all three masking levels.}
\label{tab:comparison}
\vspace{2mm}
\begin{tabular}{cccc}
\toprule
Missingness & Synthetic & Real Yaman & Difference\\
\midrule
10\%&0.539&0.470&-0.069\\
30\%&0.433&0.414&-0.019\\
50\%&0.384&0.371&-0.013\\
\bottomrule
\end{tabular}
\end{table}

\begin{figure}[h]
\centering
\includegraphics[width=0.8\linewidth]{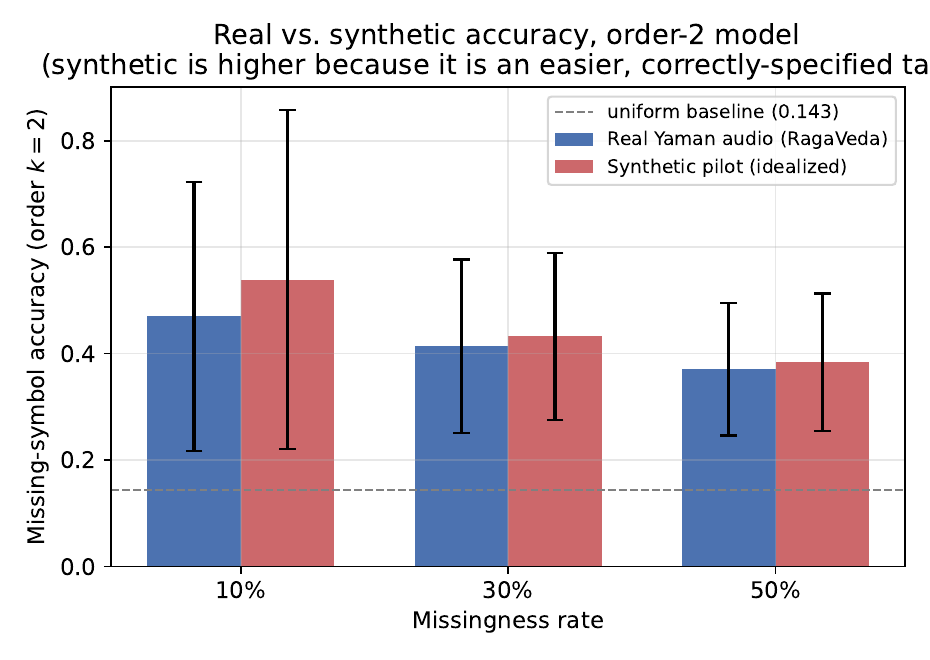}
\caption{Direct graphical comparison of order-two reconstruction accuracy. The synthetic pilot is numerically higher at every missingness level, while the real-data curve remains close to it and substantially above the uniform baseline.}
\label{fig:real-vs-synthetic}
\end{figure}

The difference shrinks from $0.069$ at $10\%$ missingness to $0.013$ at $50\%$. This is an important empirical observation: the real-data system approaches the synthetic pilot's accuracy under the most difficult masking condition despite the additional uncertainty in the real pipeline.

\subsection{Why Synthetic Data Produce Higher Raw Accuracy}
\label{ssec:comparison-why-synthetic}
There are three principal reasons. First, the synthetic generator is controlled and matched to the symbolic modeling assumptions. Second, synthetic symbols are known exactly, so there is no pitch-tracking, tonic-estimation, segmentation, or quantization error. Third, the synthetic corpus lacks natural inter-performance variability.

Thus, higher synthetic accuracy should not be interpreted as evidence that synthetic data are more representative of the intended application. It primarily reflects the fact that the synthetic task is cleaner.

\begin{figure}[h]
\centering
\includegraphics[width=0.75\linewidth]{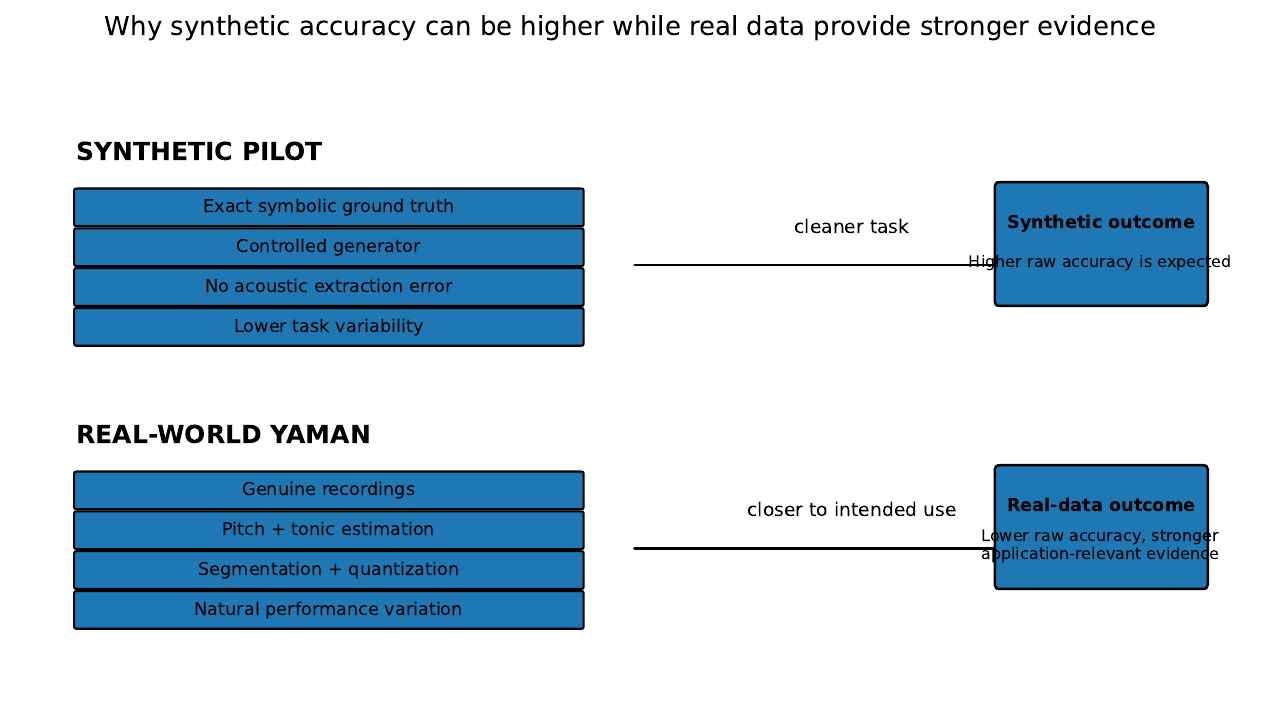}
\caption{Conceptual graphical comparison of the two evaluation regimes. Synthetic sequences begin as exact symbolic observations, whereas the real-data pipeline introduces acoustic and transcription uncertainty before reconstruction.}
\label{fig:synthetic-real-difficulty}
\end{figure}

\subsection{Why the Real-World Evaluation Is the Stronger Empirical Test}
\label{ssec:comparison-stronger}
The fact that synthetic accuracy is numerically higher does not make the synthetic experiment the stronger evidence. The synthetic experiment asks whether the implementation can recover symbols from a controlled process. The real-data experiment asks whether recoverable statistical structure survives after genuine performances have passed through an imperfect acoustic-to-symbolic pipeline.

The second question is substantially closer to the intended application of ARS. The real result therefore carries greater evidentiary weight even though its raw accuracy is lower.

\subsection{Graphical Evidence for the Real-Data Result}
\label{ssec:comparison-graphical}
Three graphical observations should be considered together. First, synthetic accuracy is higher. Second, every real-data condition remains above the uniform baseline. Third, the real data reproduce the predicted model-complexity tradeoff: order two consistently outperforms order three despite order three having a larger contextual window.

\begin{figure}[h]
\centering
\includegraphics[width=0.8\linewidth]{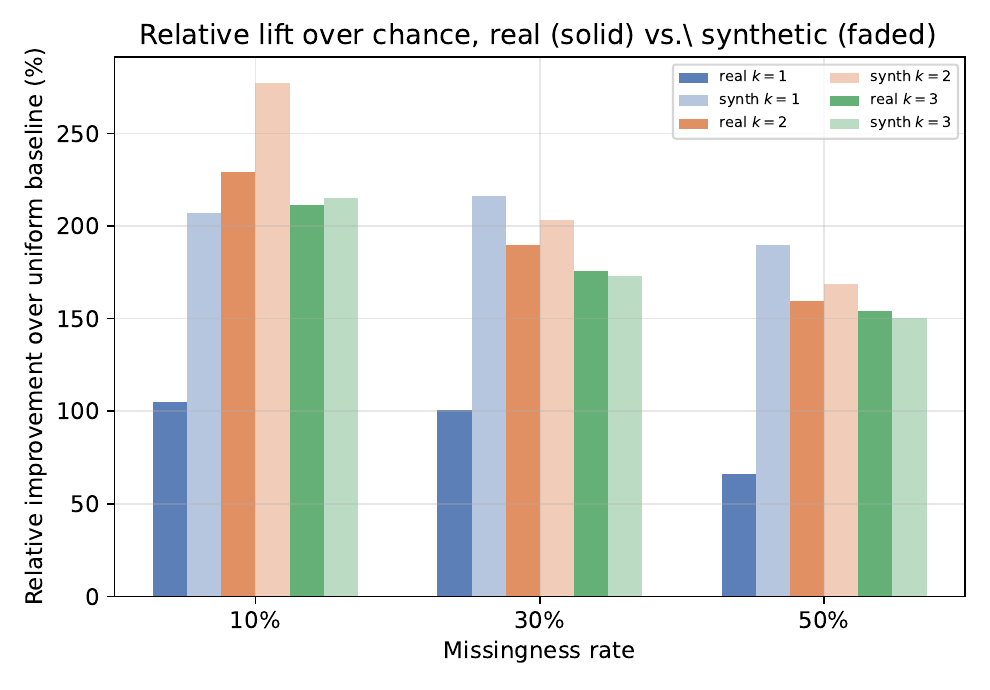}
\caption{Relative improvement over the uniform seven-symbol baseline for synthetic and real experiments. The synthetic pilot has the larger numerical lift, while the real-world experiment remains substantially above baseline across all conditions.}
\label{fig:relative-lift}
\end{figure}

This third observation links the empirical real-data results to the theoretical parameter-count analysis. For $N=7$, the parameter counts grow from $42$ to $294$ to $2058$ as order increases from one to three. The observed context coverage simultaneously falls from $100\%$ to $85.7\%$ to $62.7\%$. The empirical optimum at order two is therefore consistent with a concrete data-sparsity tradeoff rather than being a purely arbitrary numerical outcome.

\subsection{What Neither Evaluation Establishes}
\label{ssec:comparison-limits}
Neither experiment establishes musicological validity, historical authenticity, or agreement with an independently produced expert transcription. The synthetic experiment cannot establish properties of real performance because its data are generated by construction. The real experiment cannot establish expert musical correctness because its symbolic reference is generated automatically.

A stronger evaluation would require an independently annotated, properly licensed corpus, performer- or recording-disjoint splits, realistic contiguous masking and observation noise, matched baseline methods, and expert evaluation of reconstructed phrases.

\subsection{Overall Empirical Interpretation}
\label{ssec:comparison-conclusion}
The combined evidence supports a layered conclusion. The synthetic pilot validates the implementation under controlled conditions. The real-world experiment demonstrates that the same reconstruction machinery remains effective on genuine Yaman recordings. The comparison explains why the real-data accuracy is lower while showing that it remains substantially above an uninformed baseline.

Most importantly, the real-data ordering is stable across all three masking rates:
\[
0.470>0.445>0.293,
\]
\[
0.414>0.394>0.287,
\]
and
\[
0.371>0.363>0.237.
\]
Thus the real experiment does not merely produce positive accuracy; it reproduces the theoretically motivated tradeoff between contextual richness and data sparsity.

\subsection{Code Availability}
\label{ssec:code}
The reconstruction code is distributed together with the experiment configuration, fixed random seeds, figure-generation scripts, and unit tests. The repository contains the Markov estimator, exact dynamic-programming decoder, masking/evaluation loop, synthetic generator, real-audio preprocessing pipeline, and exhaustive-search verification code. The README provides the exact environment and command required to regenerate the tables and figures.

Here is the public repository for the paper: 
\begin{center}
\url{https://github.com/mathacker23/ARS_V2}
\end{center}
The repository also records the exact commit or release corresponding to the submitted manuscript. RagaVeda remains the separately cited source of the real-audio corpus.

Code availability is part of the chain of the numerical results. The paper therefore does not describe an experiment as fully reproducible unless the code, configuration, and data-processing procedure needed to reproduce it are actually available.
\newpage
\section{Algorithmic Specification}
\subsection{Training}
The training procedure consists of four logically separate stages.
\begin{enumerate}
\item Normalize the symbolic corpus and verify that every symbol belongs to $\A_R$.
\item For each candidate order $k$, count context-successor pairs $C(c,a)$.
\item Compute the Dirichlet posterior mean in equation (1).
\item Select $k$ using a training-only criterion such as BIC, then freeze the resulting model before test evaluation.
\end{enumerate}

\begin{algorithm}[h]
\caption{Training an order-$k$ raga-conditioned Markov model}
\begin{algorithmic}[1]
\Require Symbolic training corpus $D$, alphabet $\A_R$, candidate orders $K$, $\alpha>0$
\For{$k\in K$}
\State Count $C(c,a)$ for every $c\in\A_R^k$, $a\in\A_R$
\State Compute $\widehat P^{(k)}_{c,a}=(C(c,a)+\alpha)/(C(c)+N\alpha)$
\State Compute training likelihood and $\operatorname{BIC}(k)$
\EndFor
\State Select $\widehat k=\argmin_k\operatorname{BIC}(k)$
\State \Return $\widehat P^{(\widehat k)}$
\end{algorithmic}
\end{algorithm}

The algorithm does not require a neural network. This is intentional. The theorem-bearing baseline should be executable on modest hardware and should be easy for another researcher to inspect.

\subsection{Reconstruction}

Given a test fragment, the decoder initializes all admissible contexts
consistent with the observed prefix, propagates dynamic-programming scores,
rejects transitions forbidden by the grammar, and backtracks from the
highest-scoring terminal state. If no terminal state is reachable, the
decoder returns ``no feasible completion.''

\subsection{Pseudocode for exact reconstruction}

\begin{algorithm}[H]
\caption{Exact constrained reconstruction}
\begin{algorithmic}[1]

\Require Observation $O$, kernel $\widehat{P}$, constraint automaton $\mathcal{C}$

\State Initialize $V_k(z)$ for every admissible initial augmented state $z$

\For{$t=k$ to $T-1$}

    \State Set all next-layer scores to $-\infty$

    \For{each reachable state $z$}

        \For{each successor $a$ permitted by $\mathcal{C}$ and $O_{t+1}$}

            \State $z' \gets \operatorname{update}(z,a)$
            
            \State $s \gets V_t(z) + \log \widehat{P}_{c(z),a}$

            \If{$s > V_{t+1}(z')$}

                \State $V_{t+1}(z') \gets s$; store predecessor $z$

            \EndIf

        \EndFor

    \EndFor

\EndFor

\State Backtrack from the best reachable terminal state

\State \Return $\widehat{X}$ or ``no feasible completion''

\end{algorithmic}
\end{algorithm}

The exactness theorem shows that this algorithm is not a heuristic for the finite-memory model. Its limitation is computational state growth. For $N=7$ and $k=4$, there are $2401$ context states before additional grammar state is included, which is entirely manageable for short symbolic sequences. For larger $k$, longer sequences, or a rich automaton, state growth can become the dominant cost.

\subsection{When beam search is appropriate}
Beam search becomes appropriate when the score cannot be represented by a finite-memory state. A transformer decoder, for example, may condition on a continuous latent state and the entire prefix. Beam search then provides an approximation. We explicitly label it as such and evaluate sensitivity to beam width. The existence of a beam-search implementation does not alter the exactness theorem for the simpler model.

\section{Ablation and Sensitivity Analysis}
\subsection{Smoothing sensitivity}
The concentration parameter $\alpha$ controls the strength of the prior. Small $\alpha$ approaches empirical frequencies and can leave rare contexts highly variable. Large $\alpha$ pulls rows toward uniformity and can erase genuine structure. A sensitivity experiment should therefore evaluate several values, such as $0.1,0.5,1,2$, while keeping the test masks fixed.

The expected effect depends on context frequency. In a heavily sampled context, changing $\alpha$ has little effect because $C(c)\gg N\alpha$. In a rare context, the prior can dominate. This dependence is visible directly from equation (1), so the experiment can be interpreted rather than merely reported.

\subsection{Grammar sensitivity}
A hard grammar can improve accuracy if it excludes incorrect candidates without excluding the truth. It can reduce accuracy catastrophically if the encoded grammar is wrong. This asymmetry is important in cultural applications: a model that confidently enforces an incorrect rule can be less trustworthy than an unconstrained model that expresses uncertainty.

We therefore recommend a ``grammar stress test'' in which constraints are perturbed or selectively removed. Report both missing-symbol accuracy and the fraction of reconstructions that remain feasible. A model should not be praised for grammatical validity if the grammar has been made so restrictive that only one sequence is possible.

\subsection{Masking sensitivity}
Uniformly random masks are convenient but unrealistic for damaged audio. A stronger evaluation should compare isolated missing notes, contiguous gaps, burst errors, and observation noise. The exact dynamic program can handle all of these as long as the observation model is represented in the state or objective.

\subsection{Order-selection sensitivity}
Order selection should be evaluated independently of the final test score. For each training corpus, compute BIC over $k=1,2,3$ (or a larger predeclared range), record the selected order, and then evaluate that frozen model on held-out data. Reporting the distribution of selected orders across repeated training samples is more informative than reporting a single $k^*$.

A useful diagnostic is the occupancy histogram of contexts. If a high-order model contains many contexts with zero or one observation, its nominal parameter count overstates the information available for those contexts while its effective estimator variance remains high. Smoothing mitigates numerical instability but not the fundamental lack of information.

\subsection{Ablation table template}
The following table is a reporting template rather than an assertion of numerical results:
\begin{center}
\vspace{2mm}
\begin{tabular}{p{0.42\linewidth}ccc}
\toprule
Configuration & Accuracy & SD & Feasible rate\\\midrule
First-order, no grammar & -- & -- & --\\
First-order + grammar & -- & -- & --\\
Second-order + grammar & -- & -- & --\\
Third-order + grammar & -- & -- & --\\
Selected-order + smoothing & -- & -- & --\\
\bottomrule
\end{tabular}
\vspace{2mm}
\end{center}

These dashes are there for a reason and prevent fabricated values from entering a research paper. Once the repository is executed, the table can be populated automatically from the same results file used for the figures.

\subsection{Error analysis}
For every failed reconstruction, record the missing positions, local context, candidate probabilities, and the grammar rule that excluded or favored each candidate. This allows qualitative inspection of whether failures come from sparse data, incorrect grammar, ambiguous context, or genuinely non-identifiable fragments.

\section{Cultural and Scientific Limitations}
\subsection{A scale is not a raga}
The alphabet $\A_R$ is not a complete raga definition. A raga can depend on characteristic movements, phraseology, emphasis, and performance practice that are not represented by a list of pitch names. Therefore a synthetic sequence generated from a note inventory should be described as ``raga-inspired'' unless the generator has been validated against expert annotations.

\subsection{Shruti and continuous intonation}
I thought to treat the 22-shruti concept as a fixed lattice of pitch positions. That is too rigid for the present scientific claim. Historical and theoretical discussions of shruti do not justify treating every performance pitch as one of 22 globally fixed bins. The final symbolic model therefore does not force a continuous performance onto a 22-point lattice.

A future continuous model should use a tonic-normalized pitch representation and estimate uncertainty in pitch tracking. It should distinguish a measurement model from a musicological ontology: an observed frequency is an acoustic measurement, not automatically a named shruti.

\subsection{Gamaka}
Gamaka is a trajectory phenomenon. Reducing it to a single ``type'' or scalar energy discards potentially important information. A mathematically rigorous future representation should encode the curve, timing, direction, rate, and possibly performer-specific variation. Any distance between gamaka trajectories must state its invariances, such as whether time warping should count as a difference.

\subsection{Rhythm and tala}
The core model is melodic. It does not reconstruct tala, beat position, duration, or rhythmic phrasing. Since melodic and rhythmic choices can interact, a complete performance reconstruction would require a joint latent model. The absence of rhythm is a scope limitation, not evidence that rhythm is irrelevant.

\subsection{Historical authenticity}
A statistically likely completion is not necessarily the historically performed completion. This is perhaps the most important limitation. If several sequences are compatible with the observation and the chosen grammar, the model may choose one because its training corpus assigns it higher probability. That probability reflects the corpus and modeling assumptions; it does not constitute historical evidence unless independently supported.

Archival evaluation would require provenance, recording date where available, performer lineage, transcription protocol, audio-quality assessment, expert annotation, and a transparent separation between training and evaluation material. If the test recording influenced the grammar definition or model selection, the evaluation would be contaminated.

The present Yaman evaluation is deliberately weaker than that standard. Its source is a third-party RagaVeda collection without documented licensing or recording provenance, its tonic is estimated rather than annotated, and its symbolic reference is produced by the same automated pipeline used for evaluation. It should therefore be interpreted as an end-to-end feasibility check, not archival validation.

\subsection{Synthetic data}
Synthetic data are valuable because the ground truth is known. They are dangerous because the generator can make the problem easier than reality. A model trained and tested on a synthetic Markov process is naturally advantaged when the evaluation uses the same state representation. Consequently, synthetic results should be framed as implementation validation and sensitivity analysis.

\subsection{Ethical and cultural considerations}
Computational preservation should be collaborative. Musicologists, performers, archivists, and tradition bearers should participate in deciding what counts as a valid reconstruction and how uncertain outputs are presented. A system should not claim ownership over a tradition merely because it encodes a statistical representation of it. The technical framework is best understood as a tool for organizing evidence, not a replacement for practitioners.

\subsection{Reproducibility limitation}
Reproducibility needs that the repository actually contain the versioned code, configuration, environment, and exact data-generation procedure. 

\section{Discussion}
This ARS framework leads to a modest but defensible scientific claim. The core problem is mathematically tractable when melody is represented by a finite alphabet and the grammar has finite memory. Under these assumptions, the reconstruction can be solved exactly by dynamic programming. This is a meaningful result because it turns an intuitive ``fill in the missing notes'' task into a well-defined optimization problem with explicit computational complexity.

The statistical analysis also clarifies the central tradeoff. Increasing Markov order provides a richer conditional context, but the number of parameters grows as $N^k(N-1)$. In a seven-symbol alphabet, the jump from order three to order four is from 2058 to 14406 free transition parameters. Thus a higher-order model needs substantially more evidence or stronger structural assumptions.

The grammar provides a second source of information. Its benefit is strongest when the constraints are correct and informative. Its risk is model misspecification. This suggests that future systems should expose uncertainty about grammar rules rather than treating all expert-derived constraints as infallible. A Bayesian grammar model could assign posterior probabilities to competing constraint sets.

The optional HMM and VAE extensions remain scientifically useful. The HMM can capture latent phrase context, while a VAE can learn representations that are difficult to specify by hand. But the revised paper insists on a separation between representation learning and proof. A learned latent representation can improve prediction without automatically creating a theorem about grammatical validity or convergence.

The most promising direction is therefore a layered architecture: an explicit finite-state core for constraints and exact inference, augmented by learned continuous representations for acoustic detail. The interface between the two layers should be mathematically defined. For example, a neural model could propose a posterior distribution over candidate symbols while the finite-state decoder enforces hard constraints and returns the exact MAP sequence under the combined score.

The real Yaman feasibility evaluation provides a useful bridge between the theorem-bearing core and future archival work. All tested orders outperform the uniform baseline on the pipeline's own symbolic output, while order two outperforms order three at every tested missingness level. This is consistent with the predicted sparsity tradeoff, but the result should be treated as a feasibility observation rather than a general statement about Yaman or Hindustani performance. The low-confidence tonic estimates and absence of expert symbolic ground truth make the evaluation substantially weaker than an expert-annotated archival benchmark.
\newpage
\section{Conclusion}
This paper reformulates the Artificial Rosetta Stone as a principled mathematical reconstruction problem of fragmented Hindustani symbolic melodies. The main mathematical objects are a finite ragaconditional alphabet, a finite-memory constraint system, an order-$k$ Markov probability model, a Dirichlet estimator, and a constrained MAP decoder. Three results are of particular significance: 
1) The order-$k$ model has precisely $N^k(N-1)$ parameters, demonstrating the exponential cost of higher context. 

2) For fixed-length sequences with finite-memory constraints, dynamic programming is an exact optimizer of the constrained MAP objective with worst-case time complexity $O(TN^{k+1})$. 

3) Given reasonable sampling assumptions, concentration inequalities bound the error of estimated transition probabilities, while a positive score gap determines when parameter accuracy suffices to preserve the MAP sequence. Of equal significance are the claims the paper does not make. It does not suggest that a symbolic sequence completely represents a raga. It does not assert that synthetic reconstruction implies historical authenticity. 

It does not infer that a hard MAP decoder is automatically a contraction mapping. 

It does not employ an ornament-energy alignment formula as a metric without proving identity of indiscernibles and all other axioms on a clearly defined space. These limitations strengthen, rather than weaken, the research. An additional experiment the evaluation on Yaman confirms that the exact reconstruction machinery may be applied to real audio-derived symbolic sequences-order two has 0.470, 0.414, and 0.371 missing-symbol accuracy at 10\%, 30\%, and 50\% missingness, respectively, outperforming orders one and three in every condition tested. 

Given the noise in the source, tonic estimation, and symbolic ground truth, these numbers should be considered merely evidence of end-to-end feasibility and not historical or musicological validity. A follow-up study could build upon this framework by extending it to include continuous pitch, gamaka trajectories, tala, HMM phrase states, neural priors, and fully licensed archival recordings, while maintaining a rigorous mathematical foundation. The appropriate next step is a larger corpus of expert-annotated material with authentic audio degradations, cross-performer evaluation, and closely matched statistical baselines.

The Artificial Rosetta Stone should ultimately be understood as a computational inference system that organizes evidence under explicit assumptions. Its value lies not in claiming that mathematics can replace musical knowledge, but in showing exactly how formal probability, constraints, and algorithms can work alongside that knowledge without overstating what the data can prove.\newpage
\bibliographystyle{plain}
\bibliography{references.bib}

\newpage
\appendix
\section{Proofs and Technical Details}
\subsection{Proof of the parameter count}
For completeness, consider the parameter space of one transition row. It is the simplex
\[
\Delta^{N-1}=\{p\in\mathbb R^N:p_a\ge0,\ \sum_a p_a=1\}.
\]
Its affine dimension is $N-1$. The full order-$k$ kernel is the Cartesian product of $N^k$ such simplices, so its dimension is $N^k(N-1)$. This argument does not depend on the transition probabilities being strictly positive; boundary points simply correspond to zero-probability transitions.

\subsection{Proof of the Dirichlet posterior}
The Dirichlet density is proportional to $\prod_a p_a^{\alpha-1}$. The multinomial likelihood is proportional to $\prod_a p_a^{C(c,a)}$. Multiplying yields
\[
\prod_a p_a^{C(c,a)+\alpha-1},
\]
which is the kernel of a Dirichlet distribution with parameters $C(c,a)+\alpha$. The posterior mean of a Dirichlet vector with parameters $\eta_a$ is $\eta_a/\sum_b\eta_b$, giving equation (1).

The posterior mode, when all $\eta_a>1$, follows by maximizing $\sum_a(\eta_a-1)\log p_a$ subject to the simplex constraint. The Lagrangian derivative yields $p_a\propto\eta_a-1$, producing equation (2). The distinction between posterior mean and posterior mode is therefore not semantic: the formulas differ by one unit in every pseudo-count.

\subsection{Proof of KL support projection}
Let $A\subseteq\{1,\ldots,N\}$ be nonempty and let $q_a>0$. For any probability vector $p$ supported on $A$,
\[
\KL(p\|q)=\sum_{a\in A}p_a\log p_a-\sum_{a\in A}p_a\log q_a.
\]
The first term is strictly convex on the interior of the simplex, and the second is linear. Therefore the objective is strictly convex. The Lagrangian is
\[
\mathcal L(p,\eta)=\sum_{a\in A}p_a\log(p_a/q_a)+\eta\left(\sum_{a\in A}p_a-1\right).
\]
Setting derivatives to zero gives $\log(p_a/q_a)+1+\eta=0$, hence $p_a=Cq_a$. Normalization gives $C=(\sum_{b\in A}q_b)^{-1}$. Strict convexity establishes uniqueness.

This proof also exposes a boundary condition. If $q_a=0$ for an allowed $a$, then a finite KL projection can assign positive mass to that symbol only if the divergence direction is changed or smoothing has first made $q$ positive. This is another reason for keeping statistical smoothing and hard grammar conceptually separate.

\subsection{Proof of dynamic-programming optimality}
Let $S_t$ be the augmented finite state. Define $V_t(s)$ as the maximum score over all feasible prefixes ending in $s$. The Bellman optimal-substructure property follows because future score contributions depend on the past only through $s$. If two prefixes end in the same state, the one with lower score can never become optimal after the same continuation because every future continuation adds the same score to both. Therefore retaining only the best prefix per state is lossless. Repeated application of this principle gives the recurrence in equation (6), and backtracking recovers a global optimum.

\subsection{Concentration derivation}
For independent observations associated with a fixed context-symbol pair, define indicator variables $Y_i=\mathbf1\{X_i=a\}$. Then $Y_i\in[0,1]$ and $\E Y_i=p_{c,a}$. Hoeffding's inequality gives
\[
\Pp\left(\left|\frac1m\sum_{i=1}^mY_i-p_{c,a}\right|>\varepsilon\right)\le2e^{-2m\varepsilon^2}.
\]
Applying a union bound to all at most $N^{k+1}$ pairs yields equation (7). Setting the right side equal to $\delta$ and solving for $\varepsilon$ gives equation (8).

For a stationary Markov chain, the indicators are dependent. A corresponding bound can be derived using mixing inequalities, but the effective sample size is smaller than the raw number of transitions. The exact reduction depends on the mixing rate. Therefore the independent-sample result is included as a transparent baseline rather than mislabelled as a theorem for arbitrary archival sequences.

\subsection{Score stability proof}
Let $x^*$ be the unique maximizer under the true score and let $\Delta=\min_{y\ne x^*}(\ell^*(x^*)-\ell^*(y))>0$. Suppose $\sup_x|\widehat\ell(x)-\ell^*(x)|<\Delta/2$. Then for every $y\ne x^*$,
\[
\widehat\ell(x^*)-\widehat\ell(y)
>\ell^*(x^*)-\Delta/2-[\ell^*(y)+\Delta/2]
\ge0.
\]
With strict inequality in the uniform error bound, the estimated difference is positive, so $x^*$ remains the unique estimated maximizer. This establishes the proposition in Section 8.

\section{Appendix: Formal Failure Analysis of the Original Theorems}
\subsection{Failure of the stated contraction hypothesis}
Suppose $P$ is a row-stochastic matrix. By definition, $\sum_jP_{ij}=1$ for every row and $P_{ij}\ge0$. Hence
\[
\|P\|_\infty=\max_i\sum_j|P_{ij}|=1.
\]
Thus the hypothesis $\|P\|_\infty\le\tau<1$ is inconsistent with row stochasticity. This is a direct contradiction, not a subtle regularity issue.

If one instead considers a matrix $Q$ representing probability flow restricted to a proper subset of states, then $Q$ can be substochastic and may satisfy $\|Q\|_\infty<1$. But the corresponding operator is not the original transition matrix and its contraction would describe mass leakage, not necessarily convergence of a reconstruction decoder.

\subsection{Failure of argmax continuity}
Let the input space contain a scalar $u$ and two candidate outputs $x_+$ and $x_-$. Define scores $s_+(u)=u$ and $s_-(u)=-u$. The argmax output switches at $u=0$. Under the discrete metric $d(x_+,x_-)=1$, choose sequences $u_n=1/n$ and $v_n=-1/n$. Then $|u_n-v_n|=2/n\to0$ while $d(F(u_n),F(v_n))=1$. Hence $F$ is not continuous at zero and cannot be Lipschitz there.

The example is elementary, but it captures the structural issue with a MAP decoder: candidate rankings can exchange under arbitrarily small score perturbations. A contraction theorem therefore requires either a margin condition that excludes ties and near-ties or a softened operator whose output changes continuously.

\section{Appendix: Reproducibility Checklist}
Before submission, the following checklist should be satisfied.
\begin{enumerate}[leftmargin=2em]
\item The repository contains the exact code used for every numerical table and figure.
\item The random seed is fixed and printed in the experiment output.
\item Training, validation, and test sequences are generated or loaded by separate functions.
\item The test mask is generated independently of parameter fitting.
\item Any order selection is performed without using test labels.
\item Hyperparameters are listed explicitly.
\item The grammar used by the decoder is stored in a machine-readable file.
\item The README states the software environment and execution command.
\item The reported numerical values are generated from the code rather than manually copied.
\item The paper distinguishes synthetic evidence from archival evidence.
\item Confidence intervals or standard deviations are reported for repeated experiments.
\item If a public code link is unavailable, the paper says so transparently.
\end{enumerate}

\subsection{Repository structure}
\begin{verbatim}
ArtificialRosettaStone/
  README.md
  requirements.txt
  src/
    markov.py
    grammar.py
    decoder.py
    metrics.py
    generator.py
  experiments/
    config.yaml
    run_experiment.py
  results/
    results.csv
    figures/
  tests/
    test_markov.py
    test_decoder.py
\end{verbatim}

Unit tests should include probability-row normalization, Dirichlet estimates, grammar rejection, observation preservation, and comparison of dynamic programming against exhaustive enumeration on small toy cases. The last test is particularly valuable because it directly verifies the central theorem computationally: for small $N$ and $T$, exhaustive search and dynamic programming should return identical optimal scores and sequences under a deterministic tie-break rule.

\section{Appendix: Recommended Future Theoretical Work}
\subsection{Soft reconstruction operators}
A future convergence theorem could define a soft posterior operator
\[
\mathcal T_\tau(O)(x)=\frac{\exp(\ell(x;O)/\tau)}{\sum_{y\in\F(O,\R)}\exp(\ell(y;O)/\tau)},
\]
for temperature $\tau>0$. Unlike argmax, this object is a probability distribution. Under suitable influence bounds, one could investigate contraction in total variation using Dobrushin coefficients. The analysis would need to control how changing one observed symbol changes the conditional distribution over completions.

\subsection{Continuous-pitch Bayesian model}
A continuous model could define a latent pitch trajectory $\gamma_t$, an ornament process $g_t$, and an audio observation $Y$. A hierarchical factorization might take the form
\[
\Pp(Y,\gamma,g,X)=\Pp(Y\mid\gamma,g)\Pp(\gamma\mid X,g)\Pp(g\mid X)\Pp(X).
\]
This would separate the acoustic inverse problem from the symbolic grammar. Posterior inference could then combine pitch evidence with raga constraints without treating a pitch tracker as infallible.

\subsection{Uncertainty-aware grammar}
A grammar could itself be uncertain. Let $G$ be a latent grammar selected from a finite family $\{G_1,\ldots,G_J\}$. Then
\[
\Pp(X\mid O)=\sum_j\Pp(X\mid O,G_j)\Pp(G_j\mid O).
\]
This formulation prevents a single potentially incorrect rule set from producing overconfident reconstructions. It also creates a mathematically meaningful role for expert disagreement.

\subsection{Archival validation}
The present Yaman experiment is an initial feasibility bridge rather than the decisive empirical test. A stronger validating claim should use a properly licensed, attributed archival corpus where the original complete performance is known but parts are artificially masked. The masking should mimic realistic degradation, and experts should score reconstructed phrases for grammatical plausibility and stylistic fidelity. Quantitative accuracy and expert judgment should be reported separately rather than collapsed into one arbitrary score. Corpus-provided tonic and symbolic annotations should be preferred to automatic estimates where available, and performer- or recording-disjoint splits should be used to prevent leakage.

\section{Appendix: Extended Mathematical Notes}
\subsection{Initial-context estimation}
The order-$k$ likelihood requires an initial distribution $\pi$ over contexts. If the corpus supplies empirical initial contexts, a Dirichlet prior can be placed on $\pi$ in exactly the same way as for each transition row. The complete posterior then factors into the initial-context posterior and independent transition-row posteriors conditional on the observed counts.

\subsection{Stationarity is not required for the decoder}
The dynamic-programming theorem does not require the chain to be stationary. Stationarity matters for some statistical estimation arguments, not for optimization once the transition scores are fixed. This distinction is useful because musical performances may have strong nonstationary structure: opening, development, climax, and return need not share the same transition distribution.

\subsection{State augmentation}
Suppose the grammar depends on a finite phrase state $s_t$ in addition to the previous $k$ symbols. Then the augmented state is $(x_{t-k+1:t},s_t)$. If the state update is deterministic or has finitely many possibilities, dynamic programming remains exact. The complexity becomes proportional to the number of augmented states and outgoing transitions. This provides a mathematically clean bridge from the basic Markov model to an HMM-like phrase context.

\subsection{Tie-breaking}
If several candidates have equal scores, define a deterministic order on $\A_R$ and choose the lexicographically smallest maximizing sequence. This turns the set-valued MAP estimator into a function, but it does not solve the continuity problem: the selected output can still change discontinuously under arbitrarily small perturbations.

\subsection{Model misspecification}
All finite-sample guarantees above concern estimation relative to a specified model class. If the true process is not order-$k$ Markov, the estimator can converge to the best approximation in that class rather than to the true conditional distribution. This is the standard distinction between estimation error and approximation error. A complete learning-theoretic analysis should report both when possible.

\section{Appendix: Notation and Assumption Register}
For clarity, the principal assumptions used by the mathematical results are collected here. The alphabet $\A_R$ is finite with size $N$; the sequence length $T$ is fixed during reconstruction; the Markov order $k$ is finite; and every hard grammar condition used by the exact decoder has finite memory or can be represented by a finite automaton. The transition kernel is row stochastic. The Dirichlet concentration parameter satisfies $\alpha>0$.

The dynamic-programming exactness theorem is conditional only on finite-state local scoring and finite-memory constraints. It does not require stationarity, ergodicity, or a correct raga model. Those assumptions enter statistical claims instead. In particular, the concentration bound in the main text uses independent-sample reasoning. For sequential data, a mixing assumption and an effective sample size must be supplied before applying an analogous bound.

The score-stability result additionally assumes that all relevant true transition probabilities are bounded below by $p_{\min}>0$ and that the true MAP solution is separated from its nearest competitor by a positive score gap. Without a gap, arbitrarily small estimation errors can change the identity of the maximizer. This is why a parameter-estimation guarantee should never be presented as an unconditional exact-reconstruction guarantee.

The synthetic experiment assumes missing-at-random masking and a correctly specified finite symbolic generator. These assumptions are intentionally stronger than the conditions of a real archival setting. The experiment therefore validates implementation and controlled behavior, not historical authenticity. Any extension to real audio must introduce an observation model for noise, segmentation, pitch tracking, and transcription uncertainty.

Finally, the notation $R$ refers to a raga context and $\R$ to its formal constraint system only when the distinction is clear from context. All logarithms are natural logarithms. A zero probability contributes $-\infty$ to a log score and is treated as either a hard grammar exclusion or a statistically smoothed estimate according to the model specification; the two meanings are never silently conflated.

\subsection{Interpretation of guarantees}
It is useful to separate four levels of statement. At the first level, an implementation claim says that the code executes and returns a sequence. At the second level, an optimization claim says that the returned sequence is globally optimal for the mathematical objective. At the third level, a statistical claim says that the estimated objective is close to a population objective under specified assumptions. At the fourth level, a musicological claim says that the inferred sequence is a plausible or authentic representation of a musical tradition. The first three can be addressed by computer science and mathematics. The fourth requires evidence from the tradition itself.

The ARS framework in this paper establishes the second level for its finite-state core and gives conditional results at the third level. It intentionally does not collapse the fourth level into a numerical score. This distinction should remain explicit in any future presentation, abstract, poster, or software documentation. In particular, a high reconstruction accuracy on a synthetic corpus should not be described as evidence that the system has ``recovered lost music.'' It shows only that, under the chosen generator and observation process, the estimator can infer hidden symbols.

A related distinction concerns constraints. A hard grammar can guarantee that an output belongs to the encoded feasible set, but it cannot guarantee that the feasible set itself faithfully represents the tradition. In other words, constraint satisfaction is relative to the grammar. If the grammar is incomplete, an output can be mathematically valid and musically incomplete; if the grammar is wrong, an output can be mathematically valid and culturally misleading. This is why expert annotation and uncertainty-aware grammar learning are important future directions.

Finally, reproducibility should be considered part of the mathematics rather than a cosmetic appendix. A theorem is meaningful because its assumptions can be checked; a numerical result is meaningful because its computation can be repeated. The repository, configuration files, fixed seeds, and exact data-generation rules therefore form part of the evidentiary chain supporting the experimental section. The paper should be revised whenever those computational artifacts change.

\section{Appendix: Final Submission Version Notes}
The final submission should preserve the following claims exactly.

\begin{enumerate}
\item The system reconstructs a symbolic sequence under a specified probabilistic model and grammar.
\item The order-$k$ transition model has $N^k(N-1)$ free parameters before additional constraints.
\item The symmetric Dirichlet posterior mean is given by equation (1).
\item Constrained MAP reconstruction is exactly solvable by dynamic programming when the model and constraints have finite memory.
\item The worst-case dense complexity is $O(TN^{k+1})$.
\item The concentration result is conditional on its stated sampling assumptions.
\item Synthetic experiments demonstrate reproducibility and behavior under controlled missingness; they do not establish historical authenticity.
\item The personal formulation of contraction and Svara--Gamaka metric claims are deleted from inclusion in the paper as their stated proofs do not establish the conclusions.
\end{enumerate}

The strongest version of this project is the one where every definition has a clear purpose, every theorem has hypotheses that are actually satisfied, every numerical result is reproducible, and every cultural conclusion is proportional to the evidence. That standard should guide subsequent expansion to real archival recordings and continuous acoustic representations.

\subsection{Claim-to-evidence audit}
The following audit is intended to prevent the final manuscript from overstating what any component demonstrates. The symbolic reconstruction claim is supported by the explicit objective in equation (5), the feasible set definition, and the exact dynamic-programming theorem. The parameter-count claim follows algebraically from the product of simplices and does not depend on the experimental corpus. The Dirichlet estimation claim follows from conjugacy. The concentration claim is conditional on the sampling assumptions stated in Section 8 and should be weakened if the actual experiment violates them.

The computational-complexity claim is a worst-case bound. Sparse grammars can reduce the number of evaluated edges, but the bound remains a safe upper estimate for the dense order-$k$ case. If an implementation uses pruning, caching, vectorization, or parallelism, those are engineering improvements and should not be confused with a different asymptotic mathematical problem.

The cultural interpretation requires a different evidentiary standard. A raga label attached to a synthetic generator is metadata about the inspiration for the generator, not proof that the generator embodies the complete grammar of that raga. Similarly, a high symbolic accuracy score is evidence that hidden symbols were recovered under the chosen synthetic distribution; it is not evidence that a listener would regard the result as a faithful historical performance.

Here, each figure should therefore have a caption that states exactly what was measured, how many sequences were used, whether the data were synthetic or archival, and whether the displayed quantity is a mean, median, fitted parameter, or theoretical bound. Tables should identify the random seed or configuration when the values depend on stochastic computation. If a result cannot be regenerated from the repository, it should not be described as fully reproducible.

The same discipline applies to future neural extensions. A VAE reconstruction loss, attention map, or latent-space visualization may provide evidence that the model learns structure, but none is by itself a theorem of grammatical validity. If a neural decoder is constrained by a finite-state grammar, then the grammar can provide a hard validity guarantee while the neural model supplies a learned score. This division of labor is mathematically cleaner than assigning the entire guarantee to the neural architecture.

The paper is strongest when it makes this distinction explicit: probability describes what the model considers likely, constraints describe what the model permits, optimization determines which permitted candidate is selected, and empirical validation tests whether those assumptions are useful for the intended musical task. These four layers should remain separate throughout future revisions.

\vfill

\end{document}